\documentclass[letterpaper]{article}
\usepackage[T1]{fontenc}
\usepackage[utf8]{inputenc}
\usepackage[margin=1.30in]{geometry}

\usepackage{amsmath}
\usepackage{amsthm}
\usepackage{amssymb}
\usepackage[fleqn]{mathtools}
\usepackage{dsfont}
\usepackage{enumitem}
\usepackage{float}
\usepackage[noblocks]{authblk}
\usepackage{hyperref} 
\usepackage{multicol}
\usepackage{graphicx}
\usepackage{caption}
\usepackage[ruled,linesnumbered,noresetcount]{algorithm2e}
\usepackage{color}
\usepackage{soul}
\usepackage{booktabs,tabularx}
\usepackage{lscape}
\usepackage{epstopdf}

\hypersetup{
  colorlinks   = true, 
  urlcolor     = blue, 
  linkcolor    = blue, 
  citecolor   = red 
}

\SetKwRepeat{Do}{do}{while}%

\DeclareMathOperator*{\argmin}{arg\,min} 
\DeclareMathOperator*{\arginf}{arg\,inf} 

\begin{document}

\newtheorem{theorem}{Theorem}[section]
\newtheorem{lemma}[theorem]{Lemma}
\newtheorem{corollary}[theorem]{Corollary}
\newtheorem{proposition}[theorem]{Proposition}
\theoremstyle{definition}\newtheorem{definition}[theorem]{Definition}
\theoremstyle{definition}\newtheorem{problem}{Problem}
\theoremstyle{definition}\newtheorem{example}[theorem]{Example}
\theoremstyle{definition}\newtheorem{assumption}[theorem]{Assumption}
\theoremstyle{remark}\newtheorem{remark}[theorem]{Remark}
\newcommand{\eq}[1]{\begin{align}#1\end{align}}
\newcommand{\eqn}[1]{\begin{align*}#1\end{align*}}
\newcommand{\demi}{\frac{1}{2}}
\newcommand{\disp}{\displaystyle}
\newcommand{\lb}{\langle}
\newcommand{\rb}{\rangle}
\newcommand{\R}{\mathbb R}
\newcommand{\bN}{\mathbb N}
\newcommand{\bE}{\mathbb E}
\newcommand{\bZ}{\mathbb Z}
\newcommand{\bQ}{\mathbb Q}
\newcommand{\bP}{\mathbb P}
\newcommand{\bF}{\mathbb F}
\newcommand{\bG}{\mathbb G}
\newcommand{\bS}{\mathbb S}
\newcommand{\cA}{\mathcal A}
\newcommand{\cB}{\mathcal B}
\newcommand{\cF}{\mathcal F}
\newcommand{\cG}{\mathcal G}
\newcommand{\cP}{\mathcal P}
\newcommand{\cM}{\mathcal M}
\newcommand{\cN}{\mathcal N}
\newcommand{\cV}{\mathcal V}
\newcommand{\cX}{\mathcal X}
\newcommand{\dt}{\partial_t}
\newcommand{\dz}{\partial_Z}
\newcommand{\dzz}{\partial_{ZZ}}
\newcommand{\dv}{\partial_V}
\newcommand{\dvv}{\partial_{VV}}
\newcommand{\dzv}{\partial_{ZV}}
\newcommand{\Dt}{\frac{d}{dt}}
\newcommand{\Dtt}{\frac{d^2}{dt^2}}
\newcommand{\Dx}{\nabla_x}
\newcommand{\Dxx}{\nabla^2_x}
\newcommand{\EP}{\E^\bP}
\newcommand{\intT}{\int_0^T}
\newcommand{\intRd}{\int_{\R^d}}
\newcommand{\xbar}{\bar{x}}
\newcommand{\ybar}{\bar{y}}
\newcommand{\zbar}{\bar{z}}
\newcommand{\rhobar}{\bar{\rho}}
\newcommand{\etabar}{\bar{\eta}}
\newcommand{\brho}{\bar{\rho}}
\newcommand{\bsigma}{\bar{\sigma}}
\newcommand{\balpha}{\bar{\alpha}}
\newcommand{\bbeta}{\bar{\beta}}
\newcommand{\bcA}{\bar{\cA}}
\newcommand{\bcB}{\bar{\cB}}
\newcommand{\trho}{\tilde{\rho}}
\newcommand{\tcA}{\tilde{\cA}}
\newcommand{\tcB}{\tilde{\cB}}
\newcommand{\parlamphi}{\partial_{\lambda_i}\phi}
\newcommand{\abs}[1]{\lvert #1 \rvert}
\newcommand{\norm}[1]{\lVert #1 \rVert}
\newcommand{\iit}[1]{{\it #1}}
\newcommand{\bbf}[1]{{\bf #1}}
\newcommand{\braket}[2]{\langle#1,#2\rangle}

\title{Joint Modelling and Calibration of SPX and VIX by Optimal Transport}

\author[1,2]{Ivan Guo}
\author[1,2]{Gr\'egoire Loeper}
\author[3]{Jan Ob\l\'oj}
\author[1]{Shiyi Wang}
\affil[1]{School of Mathematics, Monash University, Clayton, VIC, Australia}
\affil[2]{Centre for Quantitative Finance and Investment Strategies, \protect\\ Monash University, Clayton, VIC, Australia}
\affil[3]{University of Oxford, Oxford, United Kingdom}
\date{\today}

\maketitle

\begin{abstract}
This paper addresses the joint calibration problem of SPX options and VIX options or futures. We show that the problem can be formulated as a semimartingale optimal transport problem under a finite number of discrete constraints, in the spirit of [arXiv:1906.06478]. We introduce a PDE formulation along with its dual counterpart. The solution, a calibrated diffusion process, can be represented via the solutions of Hamilton--Jacobi--Bellman equations arising from the dual formulation. The method is tested on both simulated data and market data. Numerical examples show that the model can be accurately calibrated to SPX options, VIX options and VIX futures simultaneously. \\

\end{abstract}

\section{Introduction}

The CBOE Volatility Index (VIX), also known as the stock market's ``fear gauge'', reflects the expectations of investors on the volatility of the S\&P500 index (SPX) over the next 30 days. Although the index in itself is not a tradable asset, its derivatives such as futures and options are highly liquid. Since the VIX options started trading in 2006, researchers and practitioners have been putting a lot of effort in jointly calibrating models to the SPX and VIX options prices. It has proven to be a challenging problem. As noted by many authors (e.g., \cite{jacquier2018vix, song2012tale}), inconsistencies might appear between the volatility-of-volatility inferred from SPX and VIX.

In the literature, the first attempt at jointly calibrating with continuous models\footnote{Continuous models refer to continuous-time models with continuous SPX paths.} was made by Gatheral \cite{gatheral2008consistent}, who considered a two-factor stochastic volatility model. Other attempts include a Heston model with stochastic volatility-of-volatility by Fouque and Saporito \cite{fouque2018heston} and a regime-switching stochastic volatility model by Goutte et al. \cite{goutte2017regime}. In addition, many authors have tried incorporating jumps into the SPX dynamics, see, e.g., \cite{baldeaux2014consistent,cont2013consistent,kokholm2015joint,pacati2018smiling,papanicolaou2014regime}. However, even with jumps, these models have yet to achieve satisfactory accuracy, particularly for short maturities. This leads to a natural question of whether there exists a continuous model which can capture the SPX and VIX smiles simultaneously. In \cite{acciaio2020inversion, guyon2020inversion}, Acciaio and Guyon provide a necessary condition for the existence of such continuous models. Their work was followed by the contribution of 
Gatheral et al. \cite{gatheral2020quadratic} who introduced the so-called quadratic rough Heston model that aims to provide a good approximation for both SPX and VIX smiles with only six parameters. Notably, apart from continuous models, a remarkable result was obtained by Guyon \cite{guyon2020joint} recently, who accurately reproduced the SPX and VIX smiles by modelling the distributions of SPX in discrete time. 

Recently, the theory of optimal transport was adapted to solve problems in robust hedging and pricing both in discrete and in continuous time models, see \cite{BHLP:13,HenryLabordere:2014hta}. It has proved a powerful tool since then and its applications were extended to non-parametric model calibration. In particular, the discrete-time martingale optimal transport has been applied to derive model-independent bounds on VIX derivatives by De Marco and Henry-Labordere \cite{de2015linking}. The theory has been further used to calibrate the non-parametric discrete-time model proposed by Guyon \cite{guyon2020joint}. Continuous-time optimal transport was applied by three of the authors of this paper to the calibration of local volatility \cite{guo2017local} and local-stochastic volatility models \cite{guo2019calibration} to European options. Furthermore, in \cite{guo2018path}, the first two authors have extended the semimartingale optimal transport problem \cite{tan-touzi2013} to a more general path-dependent setting. Their work expands the available calibration instruments from European options to path-dependent options, such as Asian options, barrier options and lookback options.

In this paper, we introduce a time continuous formulation of the joint calibration problem. Instead of directly modelling the instantaneous volatility of the SPX or the VIX index, we consider a semimartingale $X$ whose first element $X^1$ is the logarithm of the SPX price and whose second element $X^2$ is defined as the expectation of the forward quadratic variation of $X^1$. By doing so, the calibration exercise only depends on the marginals of $X$ at fixed times, and the joint calibration problem falls into the class of the semimartingale optimal transport problem studied in \cite{guo2019calibration}. As a corollary of the superposition principle of Trevisan \cite{Trevisan2016existence} (or earlier Figalli \cite{Figalli2008existence} for the bounded coefficients case), for any probability measure such that the drift and diffusion of $X$ are adapted processes, there exists another measure under which the semimartingale $X$ reduces to a time-inhomogeneous diffusion and has the same marginals at fixed times under both measures. It is worth noting that the idea of using diffusion processes to mimic an It\^o process by matching their marginals at fixed times traces back to the classical mimicking theorem of Gy\"ongy \cite{gyongy1986mimicking}, which was later extended by Brunick and Shreve \cite{brunick2013} to remove the conditions of nondegeneracy and boundedness on the covariance of the It\^o process. Based on this result, as shown in \cite{guo2019calibration}, it is sufficient to look for solutions among such diffusion processes. This allows us to deduce a PDE formulation of the problem along with its dual counterpart. The latter naturally gives rise to Hamilton--Jacobi--Bellman (HJB) equations which can be used to represent the solutions to the original problem. Importantly, being Markovian in the state variables, our calibrated model allows us to easily derive hedging strategies for any other options. Indeed, as long as the covariance matrix is invertible, the model is complete (see \cite{davis2008completeness}) and all derivatives based on $X$ can be fully delta hedged through dynamical trading in the SPX index and variance swaps on it.

In terms of numerical aspects, pricing of VIX derivatives involves evaluating the square root of a conditional expectation. This requires nested Monte Carlo or least square Monte Carlo methods. Nested Monte Carlo has good accuracy, but is computationally expensive. Least square Monte Carlo is efficient, but it is difficult to determine the sign of the error, which can be a useful piece of information in risk management. In the previous work of two authors of this paper \cite{guo2018pricing}, the least square Monte Carlo approach was adapted for computing the duality bounds of VIX derivatives. In this paper, by taking $X^2$ as the forward quadratic variation of $X^1$, we can use conventional Monte Carlo methods or PDE methods to calculate the prices of VIX options and futures. Then, $X$ is calibrated by a gradient descent method proposed in \cite{guo2019calibration}, in which an HJB equation is numerically solved by a fully implicit finite difference method at each iteration. It should be mentioned that a similar numerical algorithm was studied much earlier in \cite{avellaneda1997calibrating} in the context of entropy minimisation. Let us also point out that, by defining suitable state variables, our results are applicable to any calibration problem in which the calibration instruments have payoffs in the form of a function of a conditional expectation.

In fact, the calibration method presented in this paper shares many common features with Guyon's approach \cite{guyon2020joint}. For example, both methods are non-parametric and based on the theory of optimal transport, and both methods suffer from the curse of dimensionality when considering multiple maturities of VIX futures and options. Despite these similarities, there are many important differences as well. On one hand, Guyon's model is fitted to the distributions implied from market SPX and VIX options and futures, and our model is directly calibrated to the market prices of these products. On the other hand, Guyon's method seeks a three-dimensional joint probability measures on SPX and VIX at the start date of VIX and on SPX only on the end date of VIX. Our method recovers the whole trajectory distributions of SPX in a given time interval. We must acknowledge that, compared to Guyon's method, our method is more computationally expensive. We leave the study of reducing the computational complexity for future research.

The paper is organised as follows. Section 2 introduces some basic notations and the formulation of the problem. Section 3 presents the main results including a dimension reduction result, the PDE formulation and the dual formulation. Section 4 describes the numerical method in detail. Finally, in Section 5, we provide numerical examples with both simulated data and market data. 

\section{Problem formulation}
\subsection{Preliminaries}

Let $E$ be a Polish space equipped with its Borel $\sigma$-algebra. We denote $C(E)$ the set of continuous functions on $E$ and $C_b(E)$ the set of bounded continuous functions on $E$. Denote by $\cP(E)$ the set of Borel probability measures endowed with the weak-$*$ topology. Let $BV(E)$ be the set of functions of bounded variation and $L^1(d\mu)$ be the set of $\mu$-integrable functions. We also write $C(E,\R^d), C_b(E,\R^d), BV(E, \R^d)$ and $L^1(d\mu, \R^d)$ for the vector-valued versions of their corresponding sets.

Let $\Omega:=C([0,T],\R^2)$ be the two-dimensional canonical space with the canonical process $X=(X^1, X^2)$, and let $\bF=(\cF_t)_{0\leq t\leq T}$ be the canonical filtration generated by $X$. Denote by $\cP$ the set of Borel probability measures on $(\Omega, \cF_T), T>0$. Let $\cP^0\subset\cP$ denote the subset of measures such that, for each $\bP\in\cP^0$, $X\in\Omega$ is an $(\bF,\bP)$-semimartingale given by
\eq{\label{eq:semimartingale}
  X_t = X_0 + A_t + M_t, \qquad \langle X\rangle_t = \langle M\rangle_t = B_t, \quad \bP\mbox{-a.s.,}
}
where $M$ is an $(\bF,\bP)$-martingale and $(A, B)$ is $\bP$-a.s. absolutely continuous with respect to $t$. In particular, $\bP$ is said to be characterised by $(\alpha^\bP, \beta^\bP)$, which is defined in the following way,
\eqn{
  \alpha^\bP_t=\frac{dA_t}{dt},\quad \beta^\bP_t=\frac{dB_t}{dt}.
}
Note that $(\alpha^\bP,\beta^\bP)$ is $\bF$-adapted and determined up to $d\bP\times dt$, almost everywhere. In general, $(\alpha^\bP,\beta^\bP)$ takes values in the space $\R^2\times\bS^2_+$, where $\bS^2$ is the set of symmetric matrices and $\bS^2_+$ is the set of positive semidefinite matrices of order two. For any $A,B\in\bS^2$, we write $A:B=\operatorname{tr}(A^\intercal B)$. Denote by $\cP^1\subset\cP^0$ a set of probability measures $\bP$ whose characteristics $(\alpha^\bP, \beta^\bP)$ are $\bP$-integrable. In other words,
\eqn{
  \bE^\bP\left(\int_0^T\abs{\alpha_t^\bP} + \abs{\beta_t^\bP}\,dt \right) < +\infty,
}
where $|\cdot|$ is the $L^1$-norm.

Denote by $F:[0,T]\times\R^2\times\R^2\times\bS^2\to\R\cup\{+\infty\}$ a cost function, and denote by $F^*:[0,T]\times\R^2\times\R^2\times\bS^2\to\R\cup\{+\infty\}$ the convex conjugate of $F$ with respect to $(\alpha,\beta)$:
\eqn{
  F^*(t,x,a,b):=\sup_{\alpha\in\R^2,\beta\in\bS^2}\{\alpha\cdot a + \beta:b - F(t,x,\alpha,\beta) \}.
}
When there is no ambiguity, we will simply write $F(\alpha,\beta):=F(t,x,\alpha,\beta)$ and $F^*(a,b):=F^*(t,x,a,b)$.

\subsection{The joint calibration problem}

We are interested in risk-neutral measures under which the SPX price is a continuous martingale, as we assume for simplicity that both dividends and interests rates are null. Let $S_t$ be the SPX price of the form
\eqn{
  S_t = S_0 + \int_0^t\sigma_s S_s\,dW_s,
}
where $\sigma$ is some adapted process and $W$ is a one-dimensional Brownian motion.
It then follows that $X^1_t$, the logarithm of $S_t$, is a semimartingale with dynamics
\eqn{
  X^1_t=X^1_0-\demi\int_0^t\sigma_s^2\,ds + \int_0^t\sigma_s\,dW_s,\quad 0\leq t\leq T.
}
For such $X^1$, we then use $X^2$ to represent a half of the expectation of the forward quadratic variation of $X^1$ on $[t,T]$ observed at time $t$, that is
\eq{\label{eq:2-3-1}
  X^2_t = X^2_{t,T} := \bE^\bP\left(\demi\int_t^T\sigma^2_s\,ds \,\bigg|\, \cF_t\right) = X^1_t - \bE^\bP(X^1_T \mid \cF_t),\quad 0\leq t\leq T.
}
From now on, we will interchangeably use $X^2_t$ for $X^2_{t,T}$ and vice versa, $X^2_{t,T}$ being used to emphasise the dependence of $X^2$ on $T$. Note that the second term on the right-hand side of (\ref{eq:2-3-1}) is the $T$-futures price on $X^1$ at time $t$ and hence is a martingale. It follows that the modelling setting we just described is captured by probability measures $\bP\in\cP^1$ characterised by $(\alpha,\beta)$ such that
\eq{\label{eq:dynamics_characteristics}
    \alpha_t = \left[ \begin{array}{c}
-\demi\sigma_t^2 \\ -\demi\sigma_t^2\end{array} \right] \quad \mbox{and} \quad \beta_t = \left[ \begin{array}{cc}
\sigma_t^2 & (\beta_t)_{12} \\ (\beta_t)_{12} & (\beta_t)_{22} \end{array} \right], \quad 0\leq t\leq T,
}
where $(\beta_t)_{12} = d\langle X^1,X^2\rangle_t \mathbin{/} dt$ and $(\beta_t)_{22} = d\langle X^2\rangle_t \mathbin{/} dt$ and with the additional property that $X^2_{T,T}=0$ $\bP$-a.s. 

\begin{remark}
We note that this is a fully non-parametric description of all the models in $\cP^1$ compatible with the market setting described above. In particular, we do not specify the dynamics of the volatility $(\sigma_t)_{t\leq T}$. In Section \ref{sec:heston}, we show that $X$ may reproduce Heston's stochastic volatility market dynamics. More generally, we believe $X$ may capture the SPX and VIX smiles of a wide range of one-factor stochastic volatility models. However, to capture full model dynamics for other models including multi-factor stochastic volatility models, one would need to add some additional state variables so they can explicitly express $\bE^\bP(X^1_T \mid \cF_t)$ in terms of all state variables, which also increases the dimension of the problem. 
\end{remark}

In order to restrict the probability measures to those characterised by $(\alpha,\beta)$ of the form (\ref{eq:dynamics_characteristics}), we can define a cost function that penalises characteristics that are not in the following convex set:
\eqn{
  \Gamma:= \left\{(\alpha,\beta)\in\R^2\times\bS^2_+ : \alpha_1=\alpha_2=-\demi\beta_{11} \right\}.
}
Define the convex cost function $F$ as follows:
\eq{\label{eq:costfunction}
  F(\alpha,\beta) = \left\{ \begin{array}{ll} \displaystyle\sum_{i,j=1}^2 (\beta_{ij} - \bar\beta_{ij})^{2}   & \mbox{if } (\alpha,\beta)\in\Gamma, \\
		 +\infty & \mbox{otherwise,} \end{array} \right.
}
where $\bar\beta$ is a matrix of some reference values for $\beta$. Note that $\bar\beta$ may depend on $(t,X_t)$. Then, $F$ is finite if and only if $(\alpha,\beta)$ is in the form of (\ref{eq:dynamics_characteristics}). Furthermore, $F$ allows for stability across calibration exercises through specification of a reference model $\bar\beta$. 
Employing $F$ as the cost function, our aim will be to find a model which is the closest to $\bar\beta$ among the ones which calibrate fully to the given market data. We comment further on the significance of $\bar\beta$ below in Section \ref{sec:num_exp}.

The calibration instruments we consider are SPX European options, VIX options and VIX futures. The market prices of these derivatives can be imposed as constraints on $X$. Let $G$ be a vector of $m$ number of SPX option payoff functions\footnote{In the case of non-zero interest rate, the payoff functions in $G$ should be discounted.}. For example, if the $i$-th option is a put option with a strike $K_i$, then the payoff function $G_i:\R^2\to\R_+$ is given by $G_i(x) = \max(K_i-\exp(x_1), 0)$. Let $u^{SPX}\in\R^m$ be the SPX option prices and $\tau\in[0,T]^m$ be the vector of their maturities. The prices $u^{SPX}$ can be imposed on $X$ by restricting $\bP$ to probability measures that satisfy
\eqn{
  \bE^\bP G_i(X_{\tau_i})=u^{SPX}_i,\qquad \forall i=1,\ldots,m.
}

Let $0\leq t_0\leq T$. The annualised realised variance of $S_t=\exp(X^1_t)$ over a time grid $t_0<t_1<\cdots<t_n=T$ is defined to be
\eqn{
  AF\sum_{i=1}^{n}\left(\log\frac{S_{t_i}}{S_{t_{i-1}}}\right)^2,
}
where $AF$ is an annualisation factor. For example, if $t_i$ corresponds to the daily observation dates, then $AF=100^2\times252/n$, and the realised variance is expressed in basis points per annum. As $\sup_{i=1,\ldots,n}|t_i-t_{i-1}|\to 0$, the realised variance can be approximated by the quadratic variation of $X^1_t$, given by 
\eqn{
  AF\sum_{i=1}^n\left(\log\frac{S_{t_i}}{S_{t_{i-1}}}\right)^2 \overset{\bP}{\to} \frac{100^2}{T-t_0}\int_{t_0}^T\sigma_t^2\,dt.
}
The CBOE VIX index at $t_0$ is defined as the square root of a weighted average of out-of-money SPX call and put option prices with maturity $T=t_0+30$ days, which is an approximation of the implied volatility of a 30-day log-contract on the SPX. For models with continuous paths, the VIX index at $t_0$ can be expressed as the square root of the expected realised variance over the next 30 days (see \cite{dupire1993arbitrage} and \cite{neuberger1994log}), that is
\eqn{
  VIX_{t_0} &= 100\sqrt{\frac{2}{T-t_0}\bE^\bP\bigg( \demi\int_{t_0}^T\sigma_t^2\,dt \,\bigg|\, \cF_{t_0} \bigg)} = 100\sqrt{\frac{2}{T-t_0}X^2_{t_0,T}}.
}

Consider VIX options and futures both with maturity $t_0$. Let $H$ be a vector of $n$ number of VIX option payoff functions. Similarly to $G$, if the $i$-th VIX option is a put option with a strike $K_i$, then the payoff function $H_i:\R\to\R_+$ is given by $H_i(x)=\max(K_i-x,0)$. Let $J:\R^2\to\R_+$ be given by $J(x) := 100\sqrt{2x_2/(T-t_0)}$. Let $u^{VIX,f}\in\R$ be the VIX futures price and let $u^{VIX}\in\R^n$ be the VIX option prices. Then, we want to further restrict $\bP$ to those under which $X$ also satisfies the following constraints:
\eqn{
  \bE^\bP J(X_{t_0}) &= u^{VIX,f}, \\
  \bE^\bP (H_i\circ J)(X_{t_0}) &= u_i^{VIX},\qquad \forall i=1,\ldots,n.
}

Finally, to ensure that $X^2_{T,T}=0$, one additional constraint is imposed on the model. Let $\xi:\R^2\to\R_+$ be a function such that $\xi(x) = 0$ if and only if $x_2=0$. Here, we choose $\xi(x):= 1-\exp(-(x_2)^2)$ and add constraint $\bE^\bP\xi(X_T)=0$. This constraint can be interpreted as a contract that has payoff $\xi(X_T)$ at time $T$, and its price is always null. From now on, we call it the \emph{singular contract}. 

We assume that $X_0=(X^1_0, X^2_{0,T})\in\R^2$ is known, and the initial marginal of $X$ is a Dirac measure on $X_0$. The value of $X^1_0$ is the logarithm of the current SPX price. In practice, $X^2_{0,T}$ can be inferred if the market prices of SPX call and put options maturing at $T$ are available over a continuous spectrum of strikes:
\eqn{
  X^2_{0,T} = \bE^\bP\left(\demi\int_0^T\sigma^2_s\,ds\right) = \int_0^{\hat{f}_T}\frac{\bE^\bP(k-S_T)^+}{k^2}\,dk + \int_{\hat{f}_T}^\infty\frac{\bE^\bP(S_T-k)^+}{k^2}\,dk,
}
where $\hat{f}_T=\bE^\bP (S_T)$ is the $T$-forward price of the SPX index (e.g., see \cite{carr1998volatility}). If $X^2_{0,T}$ is not observable from the market, we can treat it as a parameter. Now, putting all the constraints together, we define a set of probability measures $\cP(X_0, G, H, \tau, t_0, T, u^{SPX}, u^{VIX,f},u^{VIX})\subset\cP^1$ as follows:
\eqn{
  \cP(X_0, G, H, \tau, t_0, T, u^{SPX}, u^{VIX,f},u^{VIX}):=\{ \bP\in\cP^1 : \bP\circ X_0^{-1}&=\delta_{X_0}, \\
        \bE^\bP G_i(X_{\tau_i})&=u^{SPX}_i,\, i=1,\ldots,m,\\
        \bE^\bP J(X_{t_0})&=u^{VIX,f}, \\
        \bE^\bP (H_i\circ J)(X_{t_0}) &= u_i^{VIX},\, i=1,\ldots,n,\\
        \bE^\bP\xi(X_T)&=0 \}.
}
For simplicity, we write $\cP_{joint}$ as a shorthand for $\cP(X_0, G, H, \tau, t_0, T, u^{SPX}, u^{VIX,f},u^{VIX})$. Any $\bP\in\cP_{joint}$ is a feasible risk-neutral measure under which the semimartingale $X$ reproduces the market prices. If $\cP_{joint}$ is empty, it means that the market data is not compatible with a continuous-time semimartingale model with continuous paths. Adopting the convention $\inf\emptyset=+\infty$, we formulate the joint calibration problem as a \emph{semimartingale optimal transport problem under a finite number of discrete constraints}, as studied in \cite{guo2019calibration}:
\begin{problem}\label{pb:main}
Given $X_0, G, H, \tau, t_0,,T,u^{SPX}, u^{VIX,f}$ and $u^{VIX}$, solve
\eq{\label{eq:obj_prob1}
  V:=\inf_{\bP\in\cP_{joint}} \bE^\bP\int_0^T F(\alpha_s^\bP, \beta_s^\bP)\,ds.
}
The problem is said to be \emph{admissible} if the infimum is finite and, in particular, $\cP_{joint}$ is nonempty.
\end{problem}

\begin{remark}
Let $Y$ be an $\cF_T$-measurable random variable. By identifying $X^2_t$ as a function of $X^1_t$ and $\bE^\bP(Y\mid \cF_t)$, our results apply to any model calibration problem where the payoffs of the calibration instruments can be expressed as functions of $X^1_t$ and $X^2_t$.
\end{remark}
\begin{remark}
  When considering multiple maturities for VIX futures and options, we need to have one $X^2$ for each maturity, e.g., $X^2_{t,T_1}$, $X^2_{t,T_2}$, etc. Although there is no theoretical limitation for considering multiple maturities, from numerical and practical standpoints this is challenging as each additional maturity increases the PDE's dimension.
\end{remark}

\subsection{An example: the Heston model}\label{sec:heston}

The Heston model \cite{heston1993closed} is a one-factor stochastic volatility model which directly models the spot price $S_t$ and the instantaneous variance $\nu_t$ under the risk-neutral measure. The model dynamics are given by
\eqn{
  dS_t &= \sqrt{\nu_t}S_t\,dW^1_t,  \\
  d\nu_t &= -\kappa(\nu_t - \theta)\,dt + \omega\sqrt{\nu_t}\,dW^2_t,\\
  \langle dW^1,dW^2\rangle_t &= \eta\,dt,
}
where $W^1_t$ and $W^2_t$ are standard Brownian motions with correlation $\eta$ and $\kappa, \theta > 0 $ with $2\kappa\theta>\omega^2$ so that $\nu_t>0$ a.s. In this section, we rewrite the Heston dynamics in terms of $X^1_t$ and $X^2_{t,T}$ and hence specify the probability measure $\bP\in\cP^1$ which captures the Heston dynamics. 

For $X^1$, it is obvious that $dX^1_t=d\log(S_t)=-\demi \nu_t\,dt + \sqrt{\nu_t}\,dW^1_t$. For $X^2$, by applying It\^o's formula, we have
\eq{\label{eq:heston_X2}
  X^2_{t,T} = \bE^\bP\left(\demi\int_t^T \nu_s\,ds\bigg|\cF_t\right) = \frac{1-e^{-\kappa(T-t)}}{2\kappa}(\nu_t-\theta) + \demi\theta(T-t).
}
Define $A(t,\kappa):=(1-e^{-\kappa(T-t)})/\kappa$, then a simple rearrangement of (\ref{eq:heston_X2}) gives that
\eqn{
  \nu_t = A(t,\kappa)^{-1}(2X^2_{t,T} - \theta(T-t)) + \theta =: \nu(t,X^2_{t,T},\kappa,\theta).
}
The above equation establishes a one-to-one relation between $\nu_t$ and $X^2_{t,T}$ at time $t$. Applying It\^o's formula to $X^2_{t,T}$, we have
\eqn{
  dX^2_{t,T} &= d\left(\demi A(t,\kappa)(\nu_t-\theta) + \demi\theta(T-t) \right) \\
    &= \demi(\nu_t-\theta)\,dA(t,\kappa) + \demi A(t,\kappa)\,d\nu_t - \demi\theta\,dt \\
    &= \left( \demi(\nu_t-\theta)(\kappa A(t,\kappa) - 1) - \demi\kappa A(t,\kappa)(\nu_t-\theta) - \demi\theta \right)dt + \demi A(t,\kappa)\omega\sqrt{\nu_t}\,dW^2_t \\
    &= -\demi \nu_t\,dt + \demi A(t,\kappa)\omega\sqrt{\nu_t}\,dW^2_t.
}
Therefore, the Heston model can be reformulated as
\eqn{
  dX^1_t &= -\demi \nu(t,X^2_{t,T},\kappa,\theta)\,dt + \sqrt{\nu(t,X^2_{t,T},\kappa,\theta)}\,dW^1_t, \\
  dX^2_{t,T} &= -\demi \nu(t,X^2_{t,T},\kappa,\theta)\,dt + \demi A(t,\kappa)\omega\sqrt{\nu(t,X^2_{t,T},\kappa,\theta)}\,dW^2_t,\\
  \langle dW^1_t,dW^2_t\rangle &= \eta\,dt.
}
This dynamics can be captured by the probability measure $\bP\in\cP^0$ characterised by $(\alpha,\beta)$ such that, for $t\in[0,T]$, 
\eq{\label{eq:heston_dynamics_characteristics}
    (\alpha_t,\beta_t) = \left( \left[ \begin{array}{c} -\demi \nu(t,X^2_{t,T},\kappa,\theta) \\ -\demi \nu(t,X^2_{t,T},\kappa,\theta) \end{array} \right],
    \left[ \begin{array}{cc} \nu(t,X^2_{t,T},\kappa,\theta) & \demi\eta\omega A(t,\kappa)\nu(t,X^2_{t,T},\kappa,\theta) \\ \demi\eta\omega A(t,\kappa)\nu(t,X^2_{t,T},\kappa,\theta) & \frac{1}{4}\omega^2 A(t,\kappa)^2 \nu(t,X^2_{t,T},\kappa,\theta) \end{array} \right] \right).
}
Further, it is easy to check that $\bE^\bP \int_0^T \nu(t,X^2_{t,T},\kappa,\theta)\,dt<\infty$ and hence $\bP\in \cP^1$. The characteristics (\ref{eq:heston_dynamics_characteristics}) will be used in the numerical example provided in Section \ref{sec:num_exp} for generating simulated option prices and will also be used as a reference model.

\section{Main results}

This section is devoted to presenting our main results. By following \cite{guo2019calibration}, we first present a dimension reduction result which shows that the optimal transportation cost can be achieved by a set of Markov processes. Focusing only on these Markov processes, we introduce a PDE formulation. Furthermore, we deduce a dual formulation and find the optimal characteristics as a by-product.

\subsection{Dimension reduction}

In this section, we show that if Problem \ref{pb:main} is admissible then the optimal transportation cost $V$ can be found by minimising (\ref{eq:obj_prob1}) over a subset of probability measures under which $X$ is a (time inhomogeneous) Markov processes. Before proceeding, we introduce some notations for brevity. Denote by $\bE^\bP_{t,x}$ the conditional expectation $\bE^\bP(\,\cdot\mid X_t=x)$. For any square matrix $\beta\in\bS^2_+$, we write $\beta^\demi$ such that $\beta = \beta^\demi(\beta^\demi)^\intercal$. Now, let us restate Lemma 3.1 of \cite{guo2019calibration}. 

\begin{lemma}\label{lemma:3.4}
Let $\bP\in\cP^1$ and $\rho^\bP_t=\rho^\bP(t,\cdot) = \bP\circ X_t^{-1}$ be the marginal distribution of $X_t$ under $\bP$, $t\leq T$. Then $\rho^\bP$ is a weak solution to the Fokker--Planck equation:
\eq{\label{eq:fokker-planck}
\left\{\begin{array}{r@{\ }c@{\ }l@{\ }l} \displaystyle
  \dt\rho^\bP_t + \Dx\cdot(\rho^\bP_t\bE^\bP_{t,x}\alpha^\bP_t) - \demi\sum_{i,j}\partial_{i j}(\rho^\bP_t(\bE^\bP_{t,x}\beta^\bP_t)_{ij}) &=& 0  & \quad\mbox{in }[0, T]\times\R^2, \\
  \rho^\bP_0 &=& \delta_{X_0} & \quad\mbox{in }\R^2.
\end{array}\right.
}
Moreover, there exists another probability measure $\bP'\in\cP^1$ under which $X$ has the same marginals, $\rho^{\bP'}=\rho^{\bP}$, and is a Markov process solving 
\eq{\label{eq:localised-sde}
 dX_t = \alpha^{\bP'}(t,X_t)dt + (\beta^{\bP'}(t,X_t))^\demi\,dW_t^{\bP'},\quad 0\leq t\leq T,
}
where $W^{\bP'}$ is a $\bP'$-Brownian motion, $\alpha^{\bP'}(t,X_t)=\bE^\bP_{t,X_t}\alpha^\bP_t$ and $\beta^{\bP'}(t,X_t)=\bE^\bP_{t,X_t}\beta^\bP_t$.
\end{lemma}

Lemma \ref{lemma:3.4} is a corollary of the superposition principle of Trevisan \cite{Trevisan2016existence} and Figalli \cite{Figalli2008existence}. It is worth noting that the idea of using diffusion processes to mimic an It\^o process by matching their marginals at fixed times (also called Markovian projection in the literature) traces back to the classical mimicking theorem of Gy\"ongy \cite{gyongy1986mimicking}, which was later extended by Brunick and Shreve \cite{brunick2013} to remove the conditions of nondegeneracy and boundedness on the covariance of the It\^o process.

Let $\cP_{joint}^{loc}\subset\cP_{joint}$ be the subset of probability measures under which $X$ is Markov processes in the form of (\ref{eq:localised-sde}). In other words, any $\bP'\in\cP_{joint}^{loc}$ is characterised by $(\alpha^{\bP'}(t,X_t),\beta^{\bP'}(t,X_t)):=(\bE^{\bP}_{t,X_t}\alpha^{\bP}_t,\bE^{\bP}_{t,X_t}\beta^{\bP}_t)$ for some $\bP\in\cP^1$. Moreover, under $\bP'$, $X$ has an initial marginal $\delta_{X_0}$ and is fully calibrated to the market prices given in $\cP_{joint}$. Applying Proposition 3.4 of \cite{guo2019calibration}, we have the following proposition for the joint calibration problem:

\begin{proposition}[Dimension reduction]\label{prop:localisation}
Given $\cP_{joint}$ and $\cP_{joint}^{loc}$, if Problem \ref{pb:main} is admissible, then
\eqn{
  V = \inf_{\bP\in\cP_{joint}}\bE^\bP \int_0^{T} F(\alpha_t^\bP,\beta_t^\bP) \,dt = \inf_{\bP\in\cP_{joint}^{loc}} \bE^\bP\int_0^{T} F(\alpha^\bP(t,X_t),\beta^\bP(t,X_t))\,dt.
}
\end{proposition}

\subsection{PDE formulation}

For any $\bP\in\cP_{joint}^{loc}$, the characteristics are function of the state variable $X_t$ and time $t$. As is classical in the theory of diffusions, this allows us to leverage PDE methods to describe Problem \ref{pb:main} and to use conventional numerical methods to find its solutions. 

\begin{proposition}\label{prop:pde}
If Problem \ref{pb:main} is admissible, then
\eq{\label{eq:obj_prob2}
  V=\inf_{\rho,\alpha,\beta}\int_0^{T}\int_{\R^2} F(\alpha(t,x),\beta(t,x))\,\rho(t,dx)\,dt,
}
among all $(\rho,\alpha,\beta)\in C([0,T],\cP(\R^2))\times L^1(d\rho_t dt,\R^2)\times L^1(d\rho_t dt,\bS_+^2)$ satisfying the following constraints in the sense of distributions:
\eq{
  \dt\rho(t,x) + \Dx\cdot(\rho(t,x)\alpha(t,x)) - \demi\sum_{i,j}\partial_{i j}(\rho(t,x)\beta_{ij}(t,x)) &= 0 \quad \mbox{in } [0,T]\times\R^2, \label{cons:fokker-planck}\\
  \int_{\R^2} G_i(x)\,\rho(\tau_i,dx) &= u^{SPX}_i \quad \forall i=1,\ldots,m, \label{cons:spx}\\
  \int_{\R^2} J(x)\,\rho(t_0,dx) &= u^{VIX,f}, \label{cons:vix_f}\\
  \int_{\R^2} (H_i\circ J)(x)\,\rho(t_0,dx) &= u^{VIX}_i \quad \forall i=1,\ldots,n, \label{cons:vix}\\
  \int_{\R^2} \xi(x)\,\rho(T,dx) &= 0, \label{cons:singular}
}
and the initial condition $\rho(0,\cdot) = \delta_{X_0}$.
\end{proposition}
\begin{proof}
This proposition follows immediately from Lemma \ref{lemma:3.4}. The interchange of integrals in the objective is justified by Fubini's theorem. For the weak continuity of
measure $\rho$ in time we refer the reader to \cite{loeper2006}.
\end{proof}

The PDE formulation can be solved by the alternating direction method of multipliers (ADMM) which was originally used in \cite{benamou-brenier2000} to solve the classical optimal transport. This method was extended to a one-dimensional martingale optimal transport problem in \cite{guo2017local} and to instationary mean field games with diffusion in \cite{Andreev2017mfg}. However, for problems with diffusions, the ADMM method requires to solve a fourth-order PDE with a bi-Laplacian operator. In this paper, we work on an alternative dual formulation derived by following the arguments in \cite{guo2019calibration}. This will be presented in the next subsection.

\subsection{Dual formulation}

Although the PDE formulation is not a convex problem, it can be made convex by considering the triple of measures $(\rho,\cA,\cB):= (\rho,\rho\alpha,\rho\beta)$. By doing so, the objective function (\ref{eq:obj_prob2}) is convex in $(\rho,\cA,\cB)$. Moreover, the initial condition and the constraints (\ref{cons:fokker-planck}) to (\ref{cons:singular}) are linear in $(\rho,\cA,\cB)$ and hence produce a convex feasible set. In consequence, the classical tools of convex analysis can be applied. Following Proposition 3.5 of \cite{guo2019calibration}, we introduce a dual formulation.

Let $\lambda^{SPX}\in\R^m$, $\lambda^{VIX,f}\in\R$, $\lambda^{VIX}\in\R^n$ and $\lambda^{\xi}\in\R$ be the Lagrange multipliers of constraints (\ref{cons:spx}) to (\ref{cons:singular}), respectively. To avoid confusion with the Dirac measure $\delta:\R^2\to\R\cup\{+\infty\}$ used previously, we denote by ${\cal D}:[0,T]\to\R\cup\{+\infty\}$ the Dirac delta function in the sense of distributions. The dual formulation is given as follows:
\begin{theorem}[Duality]\label{thm:duality}
If Problem \ref{pb:main} is admissible, we have
\eq{\label{eq:dual_obj_thm}
  V=\sup_{(\lambda^{SPX},\lambda^{VIX,f},\lambda^{VIX},\lambda^{\xi})\in\R^{m+n+2}} \lambda^{SPX}\cdot u^{SPX} + \lambda^{VIX,f} u^{VIX,f} + \lambda^{VIX}\cdot u^{VIX} - \phi(0, X_0),
}
where $\phi$ is the viscosity solution to the HJB equation:
\eq{\label{eq:hjb}
\begin{split}
  \dt\phi(t,x) &+ F^*(\Dx\phi(t,x),\demi\Dxx\phi(t,x)) = -\sum_{i=1}^m \lambda^{SPX}_i G_i(x) {\cal D}(t-\tau_i) \\
  &- \lambda^{VIX,f} J(x) {\cal D}(t-t_0) - \sum_{i=1}^n \lambda^{VIX}_i (H_i\circ J)(x) {\cal D}(t-t_0) - \lambda^\xi\xi(x) {\cal D}(t-T) \quad \mbox{in } [0,T]\times\R^2,
\end{split}
}
with the terminal condition $\phi(T,\cdot)=0$. Moreover, if Problem \ref{pb:main} is admissible, then the infimum in (\ref{eq:obj_prob2}) is attained. If the supremum in \eqref{eq:dual_obj_thm} is attained by some $\lambda^{SPX}$, $\lambda^{VIX,f}$, $\lambda^{VIX}$ and $\lambda^{\xi}$ for which the associated  solution to (\ref{eq:hjb}) is $\phi^*\in BV([0,T],C_b^2(\R^2))$, and if $(\rho,\alpha,\beta)$ is an optimal solution of Problem \ref{pb:main}, then $(\alpha,\beta)$ is given by
\eq{
  (\alpha_t, \beta_t) = \nabla F^*(\Dx\phi^*(t,\cdot), \demi\Dxx\phi^*(t,\cdot)),\quad d\rho_t dt-\mbox{almost everywhere.}
}
\end{theorem}

Theorem \ref{thm:duality} is an application of the Fenchel--Rockafellar duality theorem \cite[Theorem 1.9]{villani2003book}. Due to the presence of ${\cal D}$ in the source terms, the viscosity solution $\phi$ satisfies \eqref{eq:hjb} in the sense of distributions\footnote{For the precise definition of viscosity solutions to (\ref{eq:hjb}) and the corresponding comparison principle, we refer the reader to \cite[Section 3.3]{guo2019calibration}.}. Moreover, $\phi$ has possible discontinuities at $t_0$, $T$ and $\tau_i,\,i=1,\ldots,m$. The numerical solution to (\ref{eq:hjb}) is described in detail in Section \ref{sec:num_method}.  For the cost function $F$ defined in (\ref{eq:costfunction}), the convex conjugate $F^*$ is given in Lemma \ref{lemma:f_conjugate}.

\begin{remark}\label{rmk:3.5}
As mentioned in the previous work \cite{guo2019calibration}, the admissibility condition in Theorem \ref{thm:duality} was imposed for fulfilling the conditions of Fenchel--Rockafellar theorem and hence simplifying the presentation and arguments. However, it is possible to remove this assumption from Proposition \ref{prop:pde} with some modifications in the proof and still obtain the duality result in Theorem \ref{thm:duality}. Furthermore, characterising the admissibility of Problem \ref{pb:main} can be seen as a more elaborate analogue of Strassen's theorem for the classical optimal transport problem, which is however out of the scope of this paper. 
\end{remark}

In the dual formulation, the supremum can be solved by a standard optimisation algorithm. As pointed out in \cite[Lemma 4.5]{guo2019calibration}, the convergence can be improved by providing the gradients of the objective.
\begin{lemma} \label{lem:grad}
Suppose Problem \ref{pb:main} is admissible and let 
\eqn{
L(\lambda^{SPX},\lambda^{VIX,f},\lambda^{VIX},\lambda^{\xi}):= \lambda^{SPX}\cdot u^{SPX} + \lambda^{VIX,f} u^{VIX,f} + \lambda^{VIX}\cdot u^{VIX} - \phi(0, X_0).
}
Then, the gradients of the objective can be formulated as the difference between the market prices and the model prices:
\eq{
  \partial_{\lambda^{SPX}_i}L &= u^{SPX}_i - \bE^\bP G_i(X_{\tau_i}), \quad i=1,\ldots,m, \label{eq:grad_spx}\\
  \partial_{\lambda^{VIX,f}}L &= u^{VIX,f} - \bE^\bP J(X_{t_0}), \label{eq:grad_vix_f}\\
  \partial_{\lambda^{VIX}_i}L &= u^{VIX}_i - \bE^\bP (H_i\circ J)(X_{t_0}), \quad i=1,\ldots,n, \label{eq:grad_vix}\\
  \partial_{\lambda^\xi}L &= - \bE^\bP \xi(X_T). \label{eq:grad_singular}
}
\end{lemma}

In the optimisation process, the gradients are decreasing to zero while the solution is approaching the optimal solution, which illustrates the improving matching of model prices with the market prices. We note that the model prices, corresponding to a particular model $(\alpha,\beta)$, are obtained, via the Feynman-Kac formula, by solving linear pricing PDEs. More precisely, the model price of an instrument with payoff $\cal G$ and maturity $\cal T$ is equal to $\bE^\bP {\cal G}(X_{\cal T}) = \phi'(0,X_0)$, where $\phi'$ satisfies 
\eq{
  \left\{\begin{array}{l}
    \displaystyle \dt\phi' + \alpha\cdot\Dx\phi' + \demi\beta:\Dxx\phi' = 0, \qquad \mbox{in } [0,{\cal T})\times\R^2, \\
    \displaystyle \phi'({\cal T},\cdot) = {\cal G}.
  \end{array}\label{pde:pricing} \right.
}
When applying Lemma \ref{lem:grad}, we shall be using \eqref{pde:pricing} $m$ times for $({\cal G},{\cal T})=(G_i,\tau_i)$, $i=1,\ldots,m$, once for $({\cal G},{\cal T})=(J,t_0)$, $n$ times for $({\cal G}, {\cal T})=(H_i\circ J, t_0)$, $i=1,\ldots, n$, and once for $({\cal G},{\cal T})=(\xi,T)$. We shall simply refer to this as solving 
the linear pricing PDEs \eqref{pde:pricing}. Naturally, once the optimal model $(\alpha^*,\beta^*)$ is found, the above can be used not only to verify that it is indeed calibrated but also to compute other option prices under the model.

\begin{remark}
  The most computationally expensive operation of numerically solving \eqref{pde:pricing} is inverting a large sparse matrix. However, since the computations of all components of the gradient involve solving the same linear PDE but with different terminal conditions, the matrix inversion only need to be carried out once per time step. Alternatively, all gradients can be efficiently computed in one Monte Carlo simulation. In the numerical examples below (see Section \ref{sec:num_exp}), we choose to numerically solve \eqref{pde:pricing} for the sake of accuracy.
\end{remark}

\section{Numerical methods}\label{sec:num_method}
\subsection{Solving the dual formulation}

The numerical method proposed in \cite{guo2019calibration} can be directly applied to solve the dual formulation, albeit with a number of caveats. Let us first recall the numerical method. Given an initial guess $(\lambda^{SPX},\lambda^{VIX,f},\lambda^{VIX},\lambda^{\xi})$, we solve the HJB equation (\ref{eq:hjb}) to get $\phi(0,X_0)$ and hence to calculate the objective value. Due to the presence of the Dirac delta functions $\cal D$, $\phi$ might be discontinuous in time. The HJB equation can be solved in several time intervals in which, in each interval, the solution $\phi$ is continuous in both time and space, and the source terms with $\cal D$ can be incorporated into the terminal conditions. For example, if we consider SPX options with maturities $t_0$ and $T$, the HJB equation (\ref{eq:hjb}) can be reformulated as follows:
\eq{
  &\left\{\begin{array}{l}
    \begin{array}{l}
      \displaystyle \dt\phi + \sup_{\beta\in\bS^2_+}\bigg( -\demi\beta_{11}\partial_{x_1}\phi - \demi\beta_{11}\partial_{x_2}\phi + \demi\beta_{11}\partial_{x_1 x_1}\phi \\
      \displaystyle  \qquad\qquad\qquad + \beta_{12}\partial_{x_1 x_2}\phi + \demi\beta_{22}\partial_{x_2 x_2}\phi - \sum_{i,j=1}^2 (\beta_{ij} - \bar\beta_{ij})^2 \bigg) = 0 
    \end{array} \qquad\mbox{in } [t_0, T), \\
    \displaystyle \phi(T^-,\cdot) = \sum_{i=1}^m \lambda^{SPX}_i G_i \mathds{1}(\tau_i=T) + \lambda^\xi\xi, 
  \end{array}\right. \label{hjb:t0_T}\\
  &\left\{\begin{array}{l}
    \begin{array}{l}
      \displaystyle \dt\phi + \sup_{\beta\in\bS^2_+}\bigg( -\demi\beta_{11}\partial_{x_1}\phi - \demi\beta_{11}\partial_{x_2}\phi + \demi\beta_{11}\partial_{x_1 x_1}\phi \\
      \displaystyle  \qquad\qquad\qquad + \beta_{12}\partial_{x_1 x_2}\phi + \demi\beta_{22}\partial_{x_2 x_2}\phi - \sum_{i,j=1}^2 (\beta_{ij} - \bar\beta_{ij})^2 \bigg) = 0 
    \end{array} \qquad\mbox{in } [0, t_0), \\
    \displaystyle \phi(t_0^-,\cdot) = \phi(t_0,\cdot) + \sum_{i=1}^m\lambda^{SPX}_i G_i \mathds{1}(\tau_i=t_0) + \lambda^{VIX,f}J + \sum_{i=1}^n \lambda^{VIX}_i (H_i\circ J).
  \end{array}\right. \label{hjb:0_t0}
}
We then calculate the gradients of the objective by Lemma \ref{lem:grad}, in which the linear pricing PDEs \eqref{pde:pricing} are solved by an alternating direction implicit (ADI) method (see e.g., \cite{foulon2010adi}). Once we have the gradient values, we update $(\lambda^{SPX},\lambda^{VIX,f},\lambda^{VIX},\lambda^{\xi})$ by moving them against their gradients or by supplying gradients to an optimisation algorithm. Notably, the L-BFGS algorithm \cite{LBFGS1989} was employed and showed good convergence. The above steps are repeated until some optimality condition is met. When Problem \ref{pb:main} is not admissible, i.e., there does not exist a probability measure that calibrates the model to the given prices, we observe that the numerical solution will not converge, which is consistent with the arguments in Remark \ref{rmk:3.5}. The numerical method is summarised in Appendix \ref{appendix:algo}.

\subsection{Solving HJB equations}

In terms of numerical schemes for HJB equations, in their seminal work, Barles and Souganidis \cite{Barles1991monotone} have established a convergence that requires schemes to be monotone. Since then, a wide literature on monotone schemes has developed. For multidimensional HJB equations, it is usually difficult to construct a monotone scheme because of the cross partial derivative terms. To ensure monotonicity, the explicit wide stencil schemes were studied by Bonnans and Zidani \cite{Bonnans2003hjb} and by Debrabant and Jakobsen \cite{Debrabant2013semilagrangian}; however, the stability of explicit schemes are restricted by some CFL condition. In \cite{ma2017monotone}, Ma and Forsyth proposed an implicit wide stencil finite difference scheme with a local coordinate rotation which is unconditionally stable. They also maximised the use of the fixed point stencil and the central finite difference scheme to improve the order of accuracy while preserving the monotonicity of the scheme.

In this paper, we solve the HJB equations by a fully implicit finite difference method with central-difference schemes for approximating both first- and second-order derivatives. We discretise the time interval, and then, at each time step, we approximate $\beta$ by Lemma \ref{lemma:f_conjugate}. Once the optimal $\beta$ has been found, the fully nonlinear HJB equation reduces to a linear PDE which can be solved by the standard implicit finite difference method. When approximating $\beta$, we start with an arbitrary $\phi$ to approximate the derivatives of $\phi$. Next we solve the linearised PDE and plug the solution back into the supremum to approximate $\beta$ at the same time. The above procedure is repeated until $\phi$ converges, then we proceed to the next time step. This successive approximation is known as policy iteration in the literature. A good approximation to the initial $\phi$ is the one from the previous time step, which makes $\phi$ converge within a few iterations.

It is difficult to choose the boundary conditions of the HJB equations for this problem. Consider a computational domain $(x_1,x_2)\in[X^1_{min},X^1_{max}]\times[0,X^2_{max}]$. We impose the following boundary conditions to equations (\ref{hjb:t0_T}) and (\ref{hjb:0_t0}):
\eqn{
  &\left\{\begin{array}{ll}
    \Dxx\phi(t,x) = \Dxx\phi(T^-,x), \qquad &\mbox{for } (t,x)\in[t_0,T)\times(\{X^1_{min},X^1_{max}\}\times[0,X^2_{max}]\cup[X^1_{min},X^1_{max}]\times\{X^2_{max}\})  \\
    \phi(t,x) = \phi(T^-,x), \qquad &\mbox{for } (t,x)\in[t_0,T)\times[X^1_{min},X^1_{max}]\times\{0\}\\
  \end{array}\right.\\
  &\left\{\begin{array}{ll}
    \Dxx\phi(t,x) = \Dxx\phi(t_0^-,x), \qquad &\mbox{for } (t,x)\in[0,t_0)\times(\{X^1_{min},X^1_{max}\}\times[0,X^2_{max}]\cup[X^1_{min},X^1_{max}]\times\{X^2_{max}\})  \\
    \phi(t,x) = \phi(t_0^-,x), \qquad &\mbox{for } (t,x)\in[0,t_0)\times[X^1_{min},X^1_{max}]\times\{0\}\\
  \end{array}\right.
}
In addition, we set a sufficiently large computational domain to further reduce the impact of the boundary conditions. Since the linear pricing PDEs are related to the HJB equation, we use the following boundary conditions for equations \eqref{pde:pricing}:
\eqn{
  &\left\{\begin{array}{ll}
    \Dxx\phi'(t,x) = \Dxx\cG(x), \qquad &\mbox{for } (t,x)\in[0,{\cal T})\times(\{X^1_{min},X^1_{max}\}\times[0,X^2_{max}]\cup[X^1_{min},X^1_{max}]\times\{X^2_{max}\})  \\
    \phi'(t,x) = \cG(x), \qquad &\mbox{for } (t,x)\in[0,{\cal T})\times[X^1_{min},X^1_{max}]\times\{0\}.\\
  \end{array}\right.
}

As noted in \cite{Kushner2001numerical}, the standard finite difference schemes are non-monotone unless the diffusion matrix is diagonally dominated. In spite of being non-monotone in general, this scheme has the advantage of second-order accuracy for smooth solutions and ease of implementation compared to sophisticated monotone schemes. In fact, the variance of $X^2_{t,T}$ is much smaller than the variance of $X^1_t$, especially when $t$ is close to $T$. Thus, we scale up $X^2_{t,T}$ by performing a simple change of variables: $(X^1,X^2)\mapsto (X^1,KX^2)$ with $K>1$. In the numerical example of the next section we take $K=40$. Although the diffusion matrix is not diagonally dominated and the scheme is still non-monotone in general, it shows good stability and convergence for this problem after the scaling. 

\subsection{Smoothing the volatility skews}\label{sec:smoothing}

It is clear from the formulation of Problem \ref{pb:main} that the reference $\bar\beta$ influences, potentially in a very significant way, the solution. This is also confirmed by our numerics, see Section \ref{sec:simulateddata} below. However, in practice, a good selection of the reference $\bar\beta$ might not be available. Assume that there exists a $\bP_{mkt}\in\cP^{loc}_{joint}$, characterised by $(\alpha_{mkt},\beta_{mkt})$, which describes the real market dynamics. When $\bar\beta$ is far away from $\beta_{mkt}$, even though the optimised model matches all the calibrating option prices, the optimal $\beta$ may still be very different from $\beta_{mkt}$. In the numerical experiment, we observed spiky volatility surfaces and hump-shaped model volatility skews. This is not surprising because the optimiser is trying to match the model prices to the calibrating option prices while keeping $\beta$ close to $\bar\beta$.

Denote by $F^{\bar\beta}$ the cost function defined in \eqref{eq:costfunction} with reference $\bar\beta$. Let $V(\bar\beta)$ be the optimal objective value of Problem \ref{pb:main} with cost function $F^{\bar\beta}$. If $V(\bar\beta)<\infty$, by Theorem \ref{thm:duality}, $V(\bar\beta)$ is equal to the optimal objective value of the dual formulation with $(F^{\bar\beta})^*$ in the HJB equation \eqref{eq:hjb}. Let $R(\bar\beta)$ be some regularisation term that measures the smoothness of $\bar\beta$. In order to smooth out the volatility surfaces and the model volatility skews, it is natural to consider the following problem:
\eq{\label{eq:ref_iter_obj}
  \arginf_{\bar\beta\in L^1(d\rho_t dt, \bS_+^2)} V(\bar\beta) + R(\bar\beta).
} 
While we might not actually solve this problem, it motivates our \emph{reference measure iteration} method. We start with an initial reference $\bar\beta^0$ and numerically solve the dual formulation with cost function $F^{\bar\beta^0}$. Then an optimal $(\beta^*)^0$ is obtained as a by-product of solving \eqref{eq:hjb}. Next, we smooth $(\beta^*)^0$ by a simple moving average over $(t, X^1,X^2)$ with bandwidths of $(l_t, l_{x_1}, l_{x_2})$. In the numerical examples, we set $(l_t, l_{x_1}, l_{x_2}) = (3, 5, 5)$. Next, we set the smoothed $(\beta^*)^0$ to $\bar\beta^1$ and solve the dual formulation with $\bar\beta^1$. The above steps are repeated until the model volatility skews are smooth enough.

\begin{remark}
  When the calibrating instruments include VIX futures, the elements of $\bar\beta(t,x_1,x_2)$ might contain spikes around $x_2=0$, which might lead to numerical instability if we take a spiky $\bar\beta$ as the reference. In the numerical experiments below, we remove these spikes by replacing the values of $\bar\beta(\cdot,\cdot, x_2), x_2<\epsilon$ with an approximation calculated by linearly extrapolating the values of $\bar\beta(\cdot,\cdot, x_2), x_2\geq\epsilon$ along $x_2$, where $\epsilon$ is a small positive number. We find that this simple workaround effectively eliminates the numerical instability.
\end{remark}

Let us call the optimisation of solving \eqref{eq:dual_obj_thm} as the \emph{inner iteration} and call the optimisation of solving \eqref{eq:ref_iter_obj} as the \emph{outer iteration}. For the outer iteration, if the optimal $\bar\beta$ that achieves the infimum in \eqref{eq:ref_iter_obj} is not very smooth, bandwidths $(l_t, l_{x_1}, l_{x_2})$ with large values might cause the optimiser to search around the optimal $\bar\beta$ forever. Thus, $(l_t, l_{x_1}, l_{x_2})$ can be intuitively interpreted as the ``step size'' for the outer iteration. Moreover, in practice, we can apply an early stop technique by only running for a few iterations for the inner iteration. By doing so, the optimiser is alternating between the inner iteration and the outer iteration. We include this procedure in our numerical routines presented in the next section.

\section{Numerical experiments}\label{sec:num_exp}
\subsection{Simulated data}\label{sec:simulateddata}

In this section, we present a numerical example to demonstrate our method. We generate some calibrating options and futures prices from a Heston model with given parameters $(\kappa, \theta, \omega, \eta)$, and we call this model the \emph{generating model}. Next, we calibrate the semimartingale $X$ to these simulated prices by solving the dual formulation. In this case, we know that there exists such a probability measure $\bP\in\cP^1$ that $X$ can be fully calibrated to the simulated prices under $\bP$, i.e., $\cP^{loc}_{joint}\neq\emptyset$ . Recall that the interest rates and dividends are set to null. The characteristics of $\bP$ are given by \eqref{eq:heston_dynamics_characteristics} and the calibrating options and futures prices are computed by solving the linear pricing PDEs \eqref{pde:pricing}.

Recall that Problem \ref{pb:main}, combined with Proposition \ref{prop:localisation}, looks for a Markovian diffusion model which minimises a certain distance to a \emph{reference model} $\bar\beta$ subject to being calibrated. In this section we not only show that our approach is feasible but also investigate the potential influence of the choice of the reference $\bar\beta$. Specifically, we consider two reference models:
\vspace{-1em}
\eq{
  \intertext{(a) a Heston model with a different set of parameters $(\bar\kappa, \bar\theta, \bar\omega, \bar\eta)$: }
    \bar\beta(t,X^1_t,X^2_{t,T}) &= \left[ \begin{array}{cc} \nu(t,X^2_{t,T},\bar\kappa,\bar\theta) & \demi\bar\eta\bar\omega A(t,\bar\kappa)\nu(t,X^2_{t,T},\bar\kappa,\bar\theta) \\ \demi\bar\eta\bar\omega A(t,\bar\kappa)\nu(t,X^2_{t,T},\bar\kappa,\bar\theta) & \frac{1}{4}\bar\omega^2 A(t,\bar\kappa)^2 \nu(t,X^2_{t,T},\bar\kappa,\bar\theta) \end{array} \right]; \label{eq:ref_heston}
  \intertext{(b) a model with constant reference values: }
    \bar\beta(t,X^1_t,X^2_{t,T}) &= \left[ \begin{array}{cc} \bar\beta_{11} & \bar\beta_{12} \\ \bar\beta_{12} & \bar\beta_{22} \end{array} \right]. 
}
The optimal models $(\alpha^*,\beta^*)$ obtained using these two reference values will be referred to, respectively, as the \emph{OT-calibrated model with a Heston reference} and the \emph{OT-calibrated model with a constant reference}. These should not be confused with the generating (Heston) model. The idea behind the selection of candidates is to analyse the significance of $\bar\beta$ by comparing the results between two cases: (a) the dynamics of the reference model are close to the true dynamics, (b) the dynamics of the reference model are very different from the true dynamics. Note that in (a), if $(\bar\kappa, \bar\theta, \bar\omega, \bar\eta)=(\kappa, \theta, \omega, \eta)$, the supremum in \eqref{eq:dual_obj_thm} is achieved by a null vector $\mathbf{0}\in\R^{m+n+2}$ and hence $V=0$. In this case, the OT-calibrated model quickly recovers the generating model.

Let $t_0 = 49$ days and $T = 79$ days. The calibration instruments we consider are:
\begin{enumerate}
    \item SPX call options maturing at $44$ days ($= t_0 - 5$ days) and $T = 79$ days,
    \item VIX futures maturing at $t_0 = 49$ days,
    \item VIX call options maturing at $t_0 = 49$ days.
\end{enumerate}
Note that we also need to consider the singular contract (i.e., $\bE^\bP\xi(X_T)=0$) to ensure that the dynamics of $X$ are correct. All the parameter values and their interpretations are given in Table \ref{table:params}.

\begin{table}[t]
\noindent
\begin{tabularx}{\textwidth}{>{\hsize=.2\hsize\linewidth=\hsize}X
>{\hsize=.2\hsize\linewidth=\hsize}X X}
  \toprule
  Parameter & Value & Interpretation  \\
  \midrule
  $S_0$ & 100 & SPX spot price \\
  $X^1_0$ & 4.6052 & Initial position of $X^1$ \\ 
  $X^2_{0,T}$ & 0.0098 & Initial position of $X^2$ \\ 
  $\kappa$ & 0.6 & Mean reversion speed of the generating model \\
  $\theta$ & 0.09 & Long-term variance of the generating model \\
  $\omega$ & 0.4 & Volatility-of-volatility of the generating model \\
  $\eta$ & -0.5 & Correlation between SPX and variance of the generating model \\
  $\bar\kappa$ & 0.9 & Mean reversion speed of the Heston reference model \\
  $\bar\theta$ & 0.04 & Long-term variance of the Heston reference model \\
  $\bar\omega$ & 0.6 & Volatility-of-volatility of the Heston reference model \\
  $\bar\eta$ & -0.3 & Correlation between SPX and variance of the Heston reference model \\
  $\bar\beta_{11}$ & 0.09 & Reference value of $\beta_{11}$ of the constant reference model \\
  $\bar\beta_{12}$ & -0.01 & Reference value of $\beta_{12}$ of the constant reference model \\
  $\bar\beta_{22}$ & 0.04 & Reference value of $\beta_{22}$ of the constant reference model \\
  \bottomrule
\end{tabularx}
\caption{Parameter values and interpretations for the simulated data example.}
\label{table:params}
\end{table}

In this example, we consider a uniformly discretised time interval with step size $\Delta t = 0.5$ day. The numerical solutions were mainly computed on a $100\times 100$ uniform grid points, except for that we use $100\times 400$ (i.e., $400$ grid points in $X^2$) grid points for the last $10$ time steps for capturing the small variation of $X_2$ around zero when $t$ is close to $T$. 

Ideally, we want the calibrated model to have at most 1 basis point error in implied volatility for both SPX options and VIX options. However, in our method, we can only calibrate the model to option prices instead of implied volatility. Therefore, we scale the payoff functions and option prices by dividing them by their Black--Scholes vegas, which roughly converts errors in option prices to errors in implied volatility. The optimisation algorithm will iterate until the maximum error between calibrating prices and model prices are below $0.0001$, or until it cannot be further optimised. In addition, the volatility skews are smoothed by the reference measure iteration method introduced in Section \ref{sec:smoothing}.

\begin{figure}[t]
   \begin{minipage}{0.49\textwidth}
     \centering
     \includegraphics[width=1.0\linewidth]{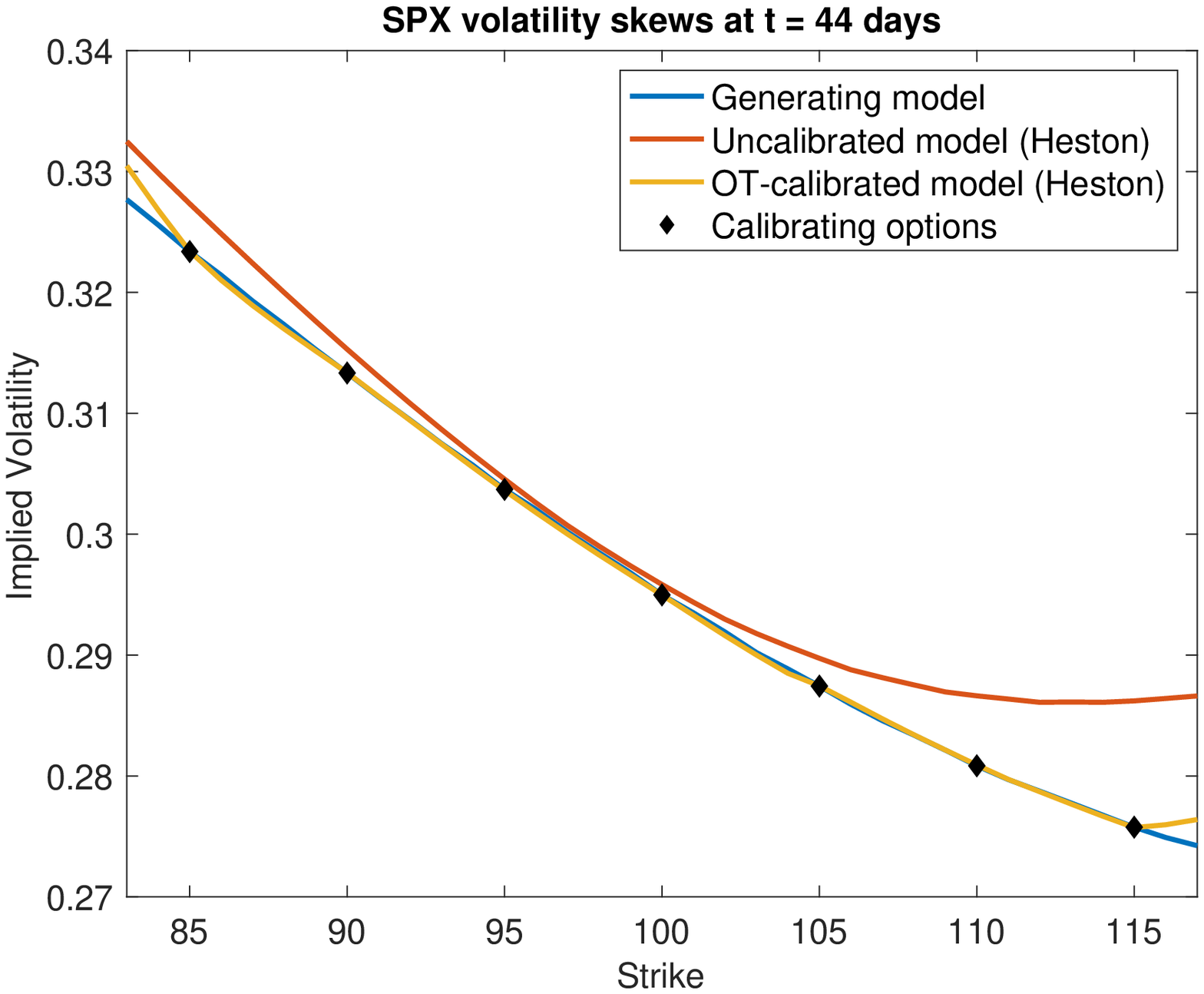}
   \end{minipage}\hfill
   \begin{minipage}{0.49\textwidth}
     \centering
     \includegraphics[width=1.0\linewidth]{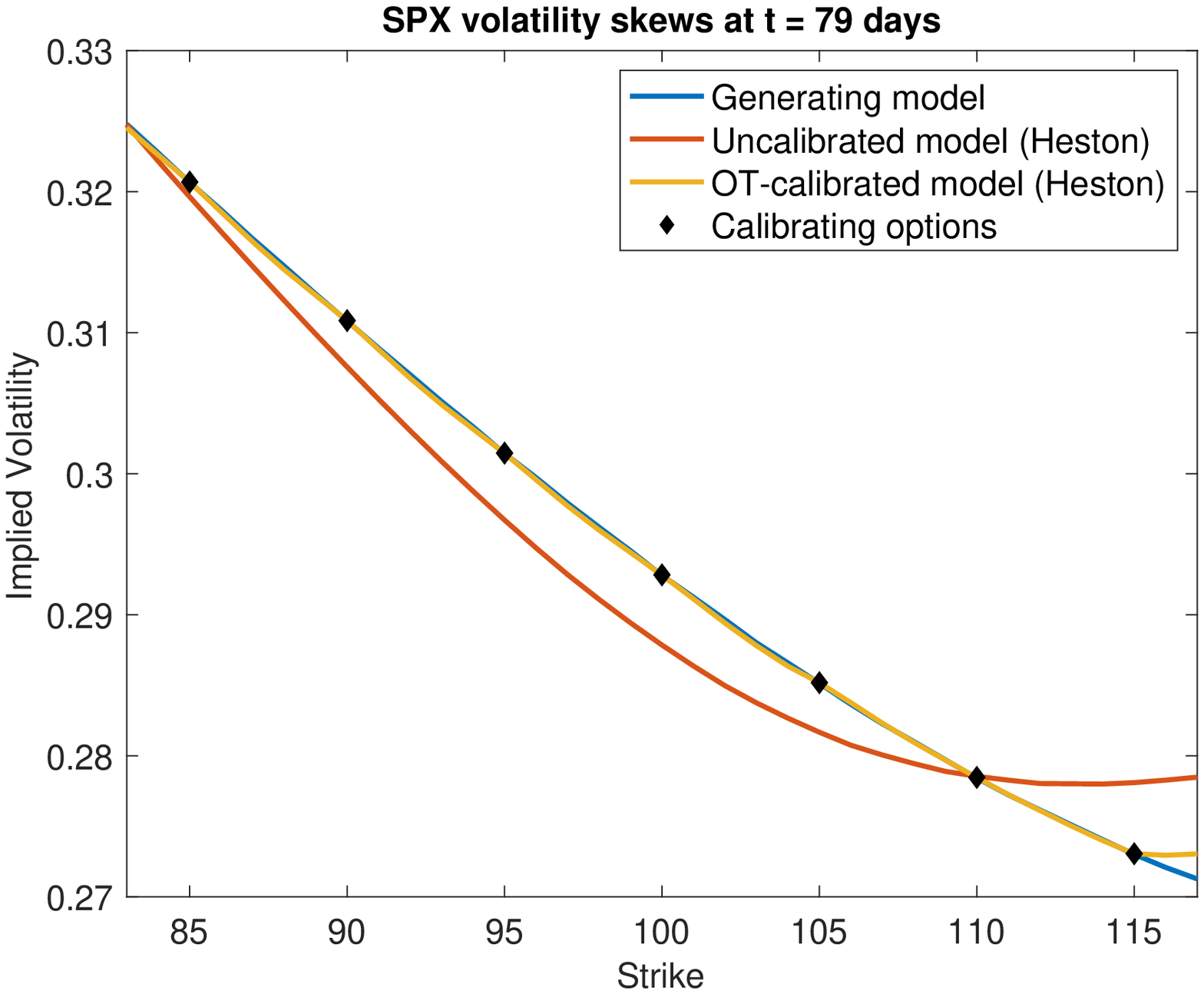}
   \end{minipage}
   \centering
   \begin{minipage}{0.49\textwidth}
     \centering
     \includegraphics[width=1.0\linewidth]{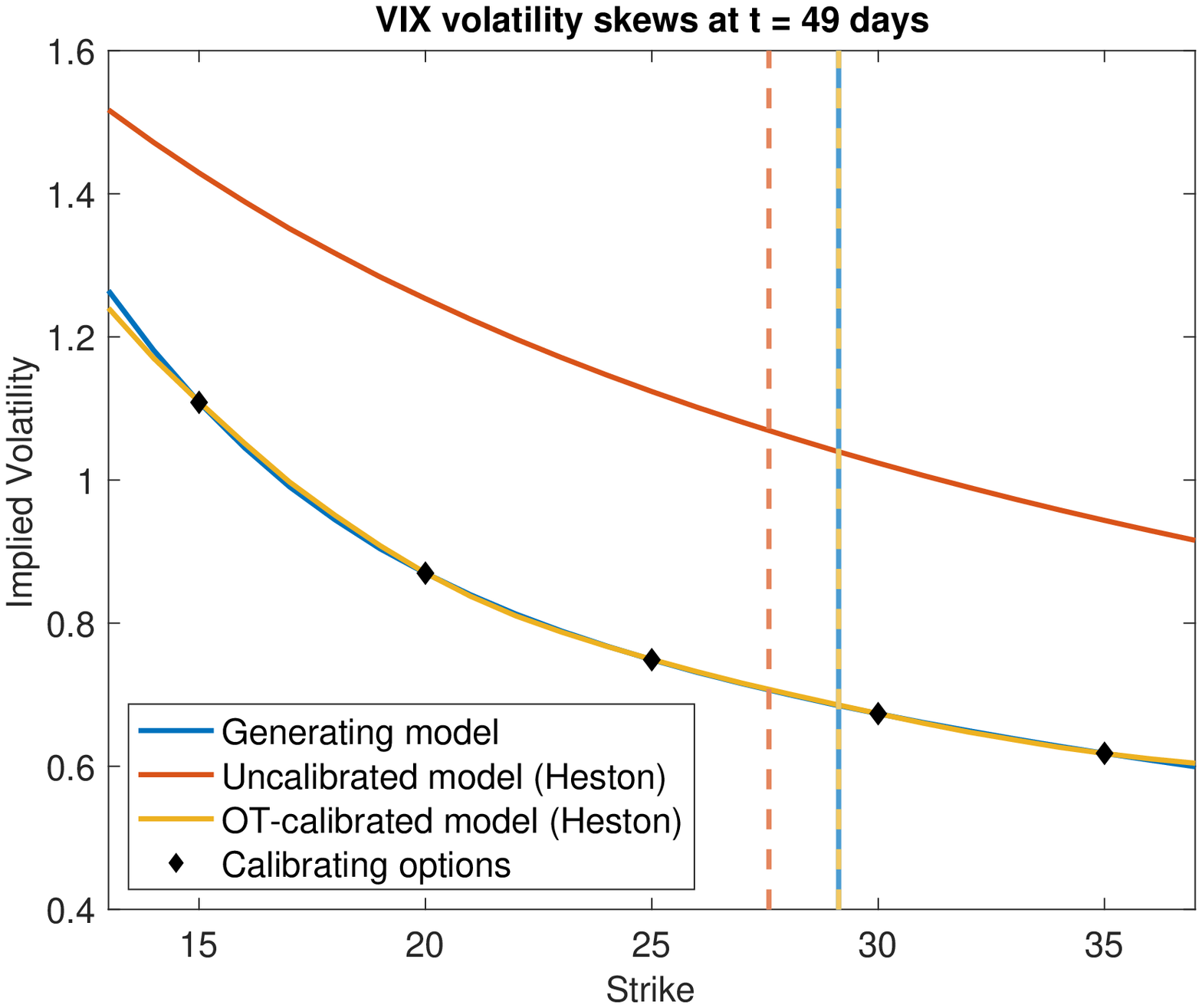}
   \end{minipage}
   \caption{The volatility skews of SPX options at $t_0-5$ days $=44$ days, SPX options at $T = 79$ days and VIX options at $t_0=49$ days for the simulated data example, including the implied volatility of the generating model, the uncalibrated Heston reference model and the OT-calibrated model with a Heston reference. The diamonds are the implied volatility of the calibrating options. The vertical lines are VIX futures prices.}
   \label{fig:vol_skews_heston}
\end{figure}

\begin{figure}[t]
   \begin{minipage}{0.49\textwidth}
     \centering
     \includegraphics[width=1.0\linewidth]{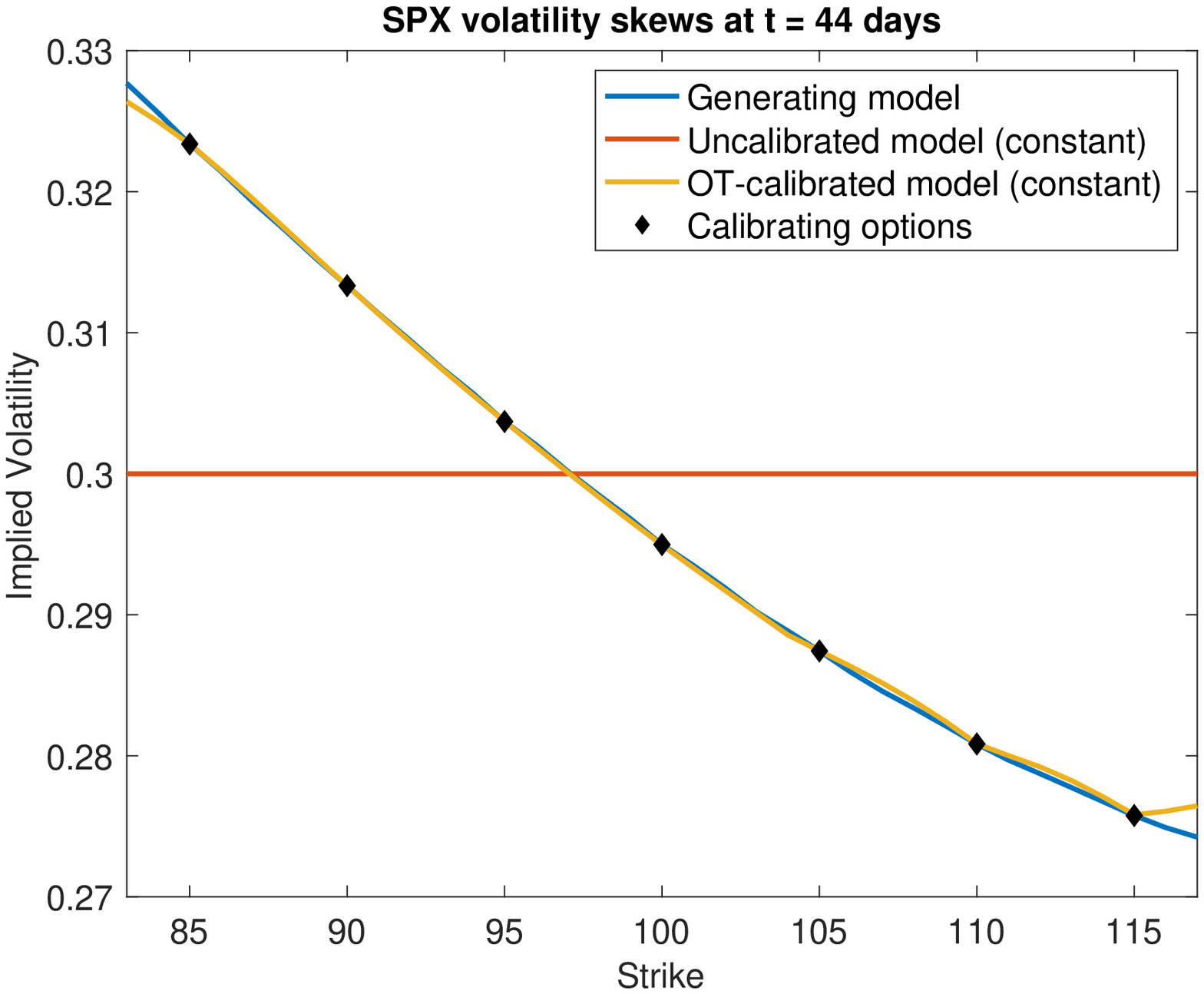}
   \end{minipage}\hfill
   \begin{minipage}{0.49\textwidth}
     \centering
     \includegraphics[width=1.0\linewidth]{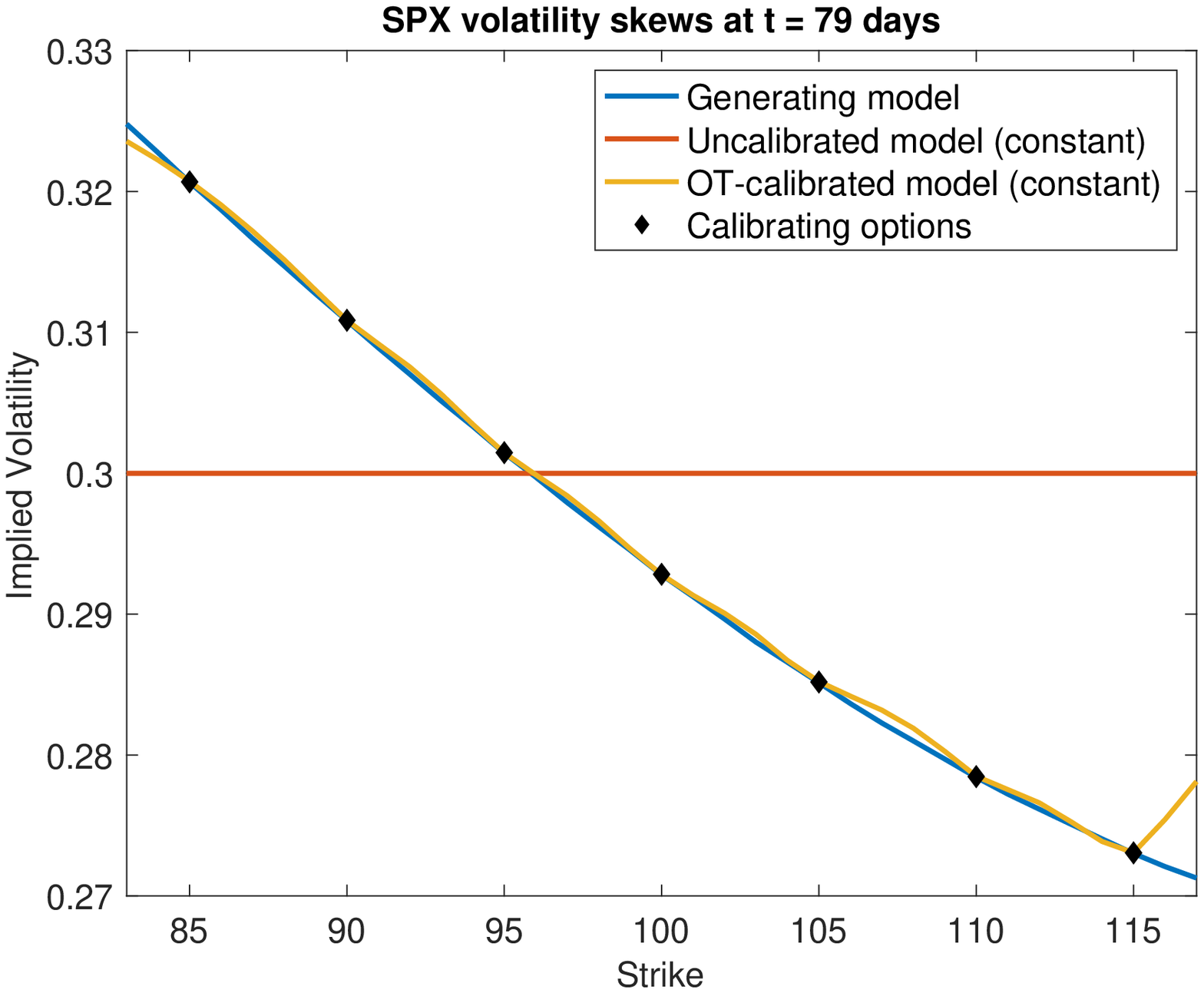}
   \end{minipage}
   \centering
   \begin{minipage}{0.49\textwidth}
     \centering
     \includegraphics[width=1.0\linewidth]{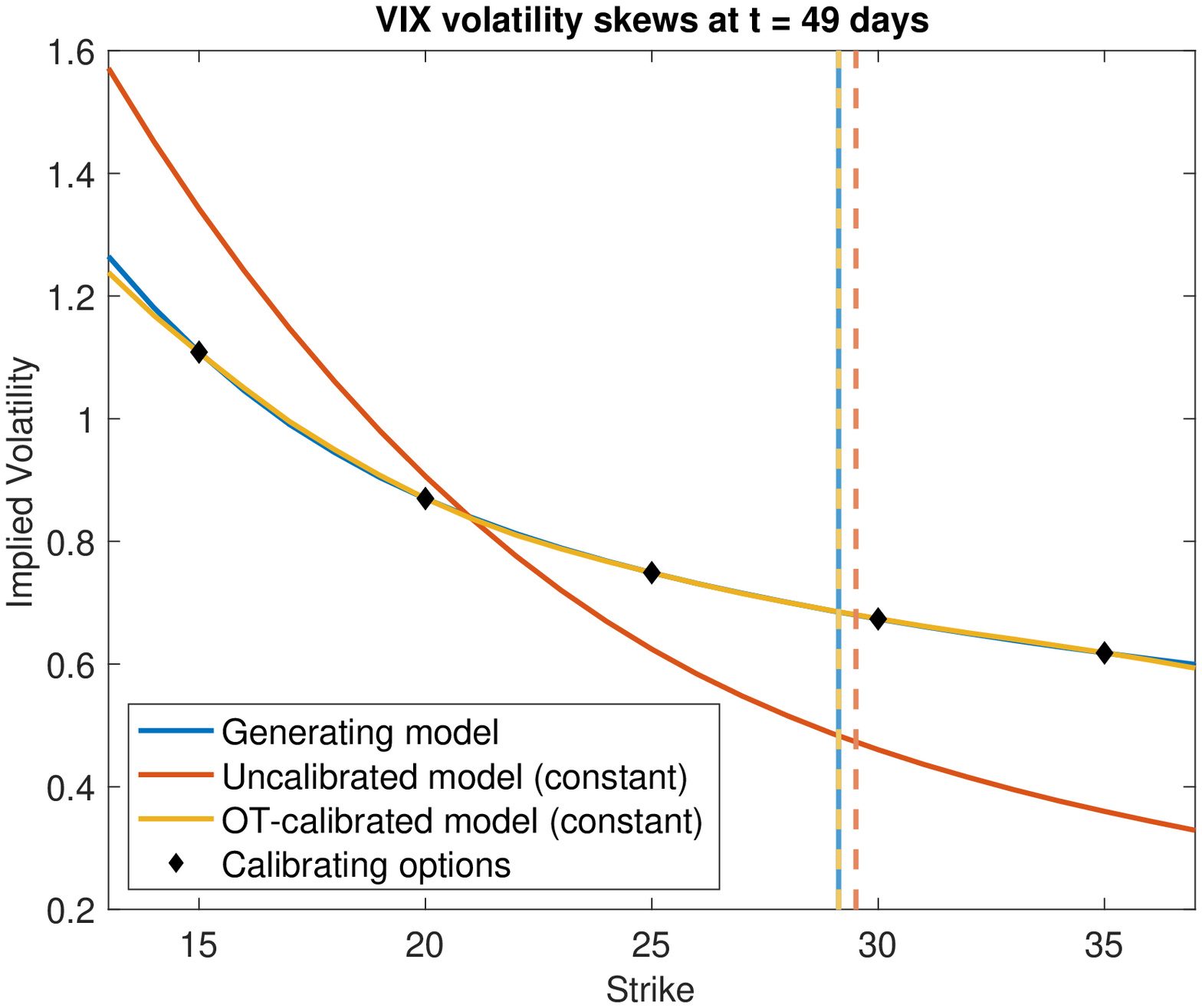}
   \end{minipage}
   \caption{The volatility skews of SPX options at $t_0-5$ days $=44$ days, SPX options at $T = 79$ days and VIX options at $t_0=49$ days for the simulated data example, including the implied volatility of the generating model, the uncalibrated constant reference model and the OT-calibrated model with a constant reference. The diamonds are the implied volatility of the calibrating options. The vertical lines are VIX futures prices.}
   \label{fig:vol_skews_constant}
\end{figure}

All numerical experiments are performed in Matlab (2020a) on a standard desktop with an i7-7700K CPU (4.5 GHz) and 32GB of RAM. The example of Heston reference takes 4 hours and the example of constant reference takes 10.7 hours. The reason that the latter example takes longer to complete is that as the constant reference value is very different from the generating model, it takes more iterations to smooth the volatility surfaces and skews by using the reference measure iteration method. We must acknowledge that our method is very computationally expensive. We plan to study on reducing the computational time in future research.

The calibration results are shown in Table \ref{table:result}, and the volatility skews are given in Figure \ref{fig:vol_skews_heston}--\ref{fig:vol_skews_constant}. We can see that the OT-calibrated models, both with the Heston reference and the constant reference, accurately capture the calibrating SPX options, VIX futures and VIX options prices. The errors, in implied volatility, of the SPX options are at most 1 basis point and of the VIX options are at most 10 basis points.

To verify if the model dynamics are correct, we perform a Monte Carlo simulation of $X$ with the Euler scheme, and the results are shown in Figure \ref{fig:sim_x1}--\ref{fig:sim_x2}. As demonstrated, $X^2_{T,T} \approx 0$ in all three models, so we consider the constraints $X^2_{T,T}=0$ $\bP$-a.s. are satisfied, and the model dynamics are correct. 

Regarding the robustness of the method, there is no doubt that the reference value has a significant influence on the model dynamics. In Figures \ref{fig:vol_skews_heston} and \ref{fig:vol_skews_constant}, the SPX model volatility skews show some differences between the ones with different reference values. In the intervals between any two adjacent option strikes, these difference are relative small, which is the result of the smoothing method. In the intervals that are less than the smallest strike and greater than the largest strike, these differences are relative large, because the model is penalised away from the reference values. Surprisingly, the VIX model volatility skews show only small differences. In Figures \ref{fig:sim_x1} and \ref{fig:sim_x2}, we note that the dynamics of the three models are different. In fact, the OT-calibrated model with the constant reference is very different from the other two models. We further display the volatility behaviour of the three models in Appendix \ref{app:beta}.

\begin{figure}[t]
\begin{center}
\includegraphics[width=\textwidth]{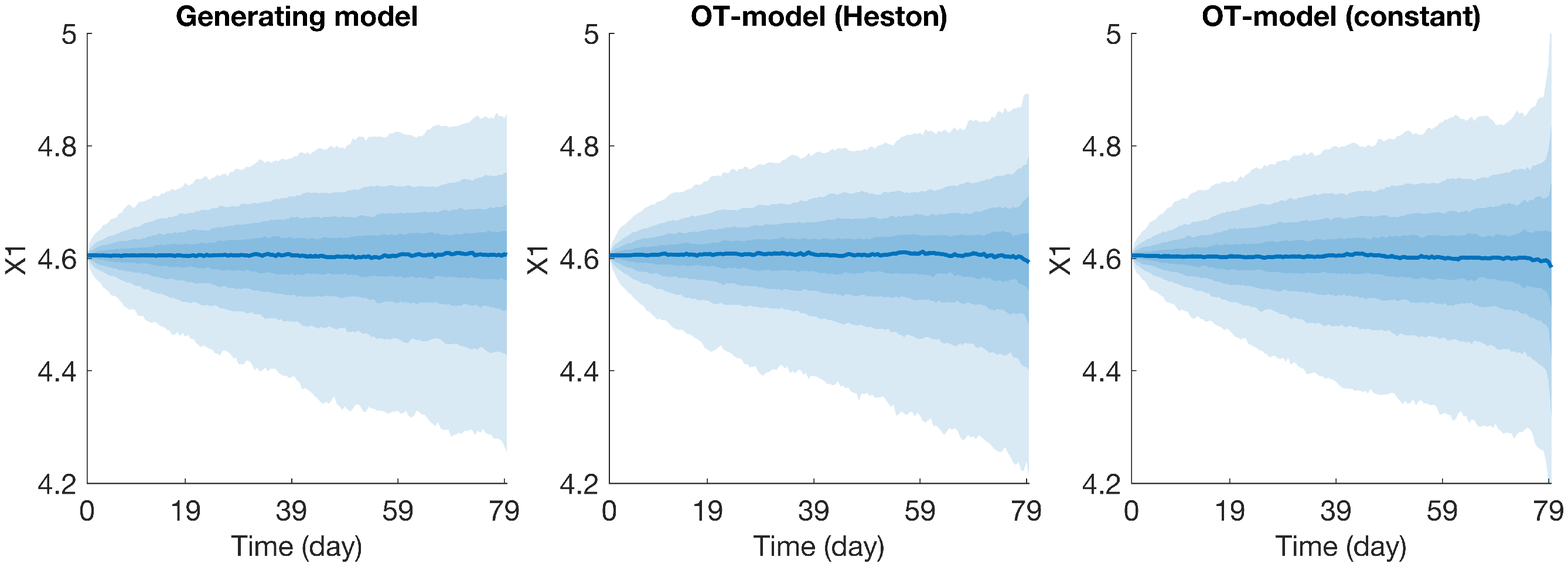}
\caption{The simulations of $X^1_t$ for the simulated data example, including the generating model, the OT-calibrated model with a Heston reference and the OT-calibrated model with a constant reference.}
\label{fig:sim_x1}
\end{center}
\end{figure}

\begin{figure}[t]
\begin{center}
\includegraphics[width=\textwidth]{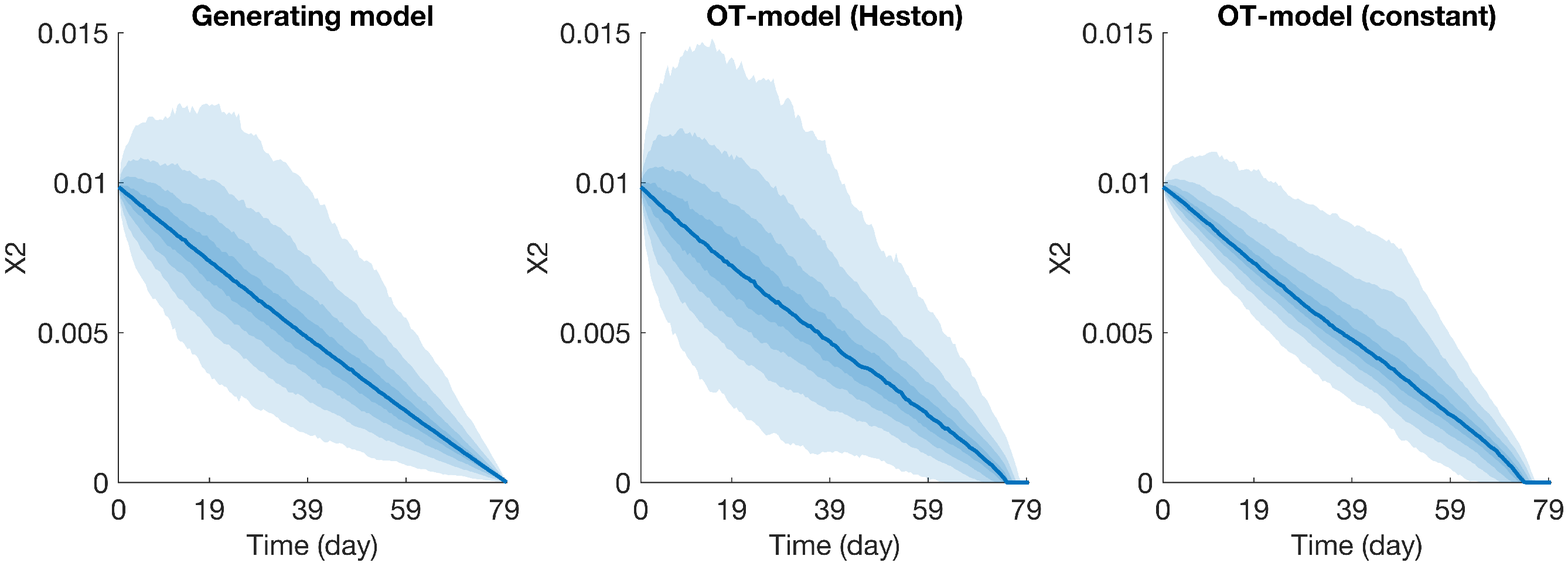}
\caption{The simulations of $X^2_{t,T}$ for the simulated data example, including the generating model, the OT-calibrated model with a Heston reference and the OT-calibrated model with a constant reference.}
\label{fig:sim_x2}
\end{center}
\end{figure}

\begin{landscape}
\begin{table}[p]
\centering
\resizebox{\hsize}{!}{%
\begin{tabular}{ccccccccc}
  \toprule
 & & & \multicolumn{2}{c}{Generating model} & \multicolumn{2}{c}{OT-model (Heston)} & \multicolumn{2}{c}{OT-model (constant)}  \\
 \cmidrule(lr){4-5} \cmidrule(lr){6-7} \cmidrule(lr){8-9}
 & Maturity & Strike & Price & IV & Model price & Model IV & Model price & Model IV  \\
  \midrule
 SPX call options & 44 days & 85 & 15.3513 & 0.3234 & 15.3514 (0.0001) & 0.3234 (0.0000) & 15.3512 (-0.0001) & 0.3234 (0.0000)  \\
 &  & 90 & 10.9298 & 0.3133 & 10.9300 (0.0002) & 0.3134 (0.0001) & 10.9297 (-0.0001) & 0.3133 (0.0000)  \\
 &  & 95 & 7.0999 & 0.3037 & 7.0989 (-0.0010) & 0.3036 (-0.0001) & 7.1000 (0.0001) & 0.3037 (0.0000)  \\
 &  & 100 & 4.1123 & 0.2950 & 4.1121 (-0.0002) & 0.2950 (0.0000) & 4.1118 (-0.0005) & 0.2949 (-0.0001)  \\
 &  & 105 & 2.0817 & 0.2874 & 2.0819 (0.0002) & 0.2875 (0.0001) & 2.0818 (0.0001) & 0.2874 (0.0000)  \\
 &  & 110 & 0.9061 & 0.2808 & 0.9068 (0.0007) & 0.2809 (0.0001) & 0.9063 (0.0002) & 0.2809 (0.0001)  \\
 &  & 115 & 0.3392 & 0.2758 & 0.3390 (-0.0002) & 0.2757 (-0.0001) & 0.3395 (0.0003) & 0.2758 (0.0000)  \\
 \cmidrule{2-9}
 & 79 days & 85 & 15.9829 & 0.3207 & 15.9832 (0.0003) & 0.3207 (0.0000) & 15.9836 (0.0007) & 0.3207 (0.0000)  \\
 &  & 90 & 11.8931 & 0.3108 & 11.8936 (0.0005) & 0.3109 (0.0001) & 11.8934 (0.0003) & 0.3108 (0.0000)  \\
 &  & 95 & 8.3453 & 0.3014 & 8.3457 (0.0004) & 0.3015 (0.0001) & 8.3456 (0.0003) & 0.3014 (0.0000)  \\
 &  & 100 & 5.4675 & 0.2928 & 5.4680 (0.0005) & 0.2928 (0.0000) & 5.4678 (0.0003) & 0.2928 (0.0000)  \\
 &  & 105 & 3.3174 & 0.2851 & 3.3182 (0.0008) & 0.2852 (0.0001) & 3.3188 (0.0014) & 0.2852 (0.0001)  \\
 &  & 110 & 1.8524 & 0.2784 & 1.8529 (0.0005) & 0.2785 (0.0001) & 1.8535 (0.0011) & 0.2785 (0.0001)  \\
 &  & 115 & 0.9533 & 0.2730 & 0.9539 (0.0006) & 0.2731 (0.0001) & 0.9539 (0.0006) & 0.2731 (0.0001)  \\
 \midrule
 VIX call options & 49 days & 15 & 14.3139 & 1.1086 & 14.3146 (0.0007) & 1.1094 (0.0008) & 14.3131 (-0.0008) & 1.1076 (-0.0010)  \\
 &  & 20 & 9.5850 & 0.8699 & 9.5856 (0.0006) & 0.8702 (0.0003) & 9.5854 (0.0004) & 0.8701 (0.0002)  \\
 &  & 25 & 5.4779 & 0.7489 & 5.4794 (0.0015) & 0.7494 (0.0005) & 5.4778 (-0.0001) & 0.7489 (0.0000)  \\
 &  & 30 & 2.5079 & 0.6735 & 2.5085 (0.0006) & 0.6737 (0.0002) & 2.5102 (0.0023) & 0.6741 (0.0006)  \\
 &  & 35 & 0.8639 & 0.6181 & 0.8632 (-0.0007) & 0.6179 (-0.0002) & 0.8652 (0.0013) & 0.6185 (0.0004)  \\
 \midrule
 VIX futures & 49 days &  & 29.1285 &  & 29.1292 (0.0007) &  & 29.1268 (-0.0017) &   \\
 \midrule
 Singular contract & 79 days &  & 0 &  & 5.34E-06 &  &  5.26E-08 &   \\
 \bottomrule
\end{tabular}
}
\caption{The calibration results of the simulated data example, including prices and implied volatility (IV) of the generating model, the OT-calibrated model with a Heston reference and the OT-calibrated model with a constant reference. The errors are shown in the parentheses. }
\label{table:result}
\end{table}
\end{landscape}

\subsection{Market data}\label{sec:marketdata}

To further test the effectiveness of our method, we calibrate the model to the market data as of September 1st, 2020. 
\begin{remark}
For simplicity, we have assumed that the interest rates and dividends are null, and the spot price is a martingale under the risk-neutral measure. However, this assumption does not apply to the market data. To overcome this issue, we let $X^1$ be the logarithm of the T-forward price of the SPX index instead of the spot price. Then, we are interested in T-forward measures $\bP\in\cP^1$ under which $\exp(X^1)$ is a martingale.
\end{remark}
The market data consists of monthly SPX options maturing at 17 days and 45 days and monthly VIX futures and options maturing at 15 days. The model is optimised with a Heston reference \eqref{eq:ref_heston} with parameters given in Table \ref{table:params_market}. The parameters are obtained by (roughly) calibrating a standard Heston model to the SPX option prices. It should be noted that, even with these parameters, the VIX skew generated by the Heston reference model is very unrealistic. 
Numerically, we have also observed that the convergence is sensitive to $\bar\beta$. Therefore, we apply the reference measure iteration method, developed in Section \ref{sec:smoothing}, to iteratively improve the reference value. The total computation time (including the reference measure iterations) is 11 hours. From a practical perspective, one way to reduce the computation time is to set the reference value to a pre-calibrated $\beta$. Nevertheless, we leave the task of finding better reference values and reducing the computation time for future research.

\begin{table}[h]
\noindent
\begin{tabularx}{\textwidth}{XXXXXXX}
  \toprule
  Parameter & $X^1_0$ & $X^2_{0,T}$ & $\bar\kappa$ & $\bar\theta$ & $\bar\omega$ & $\bar\eta$  \\
  \midrule
  Value & 8.17 & 0.0048 & 4.99 & 0.038 & 0.52 & -0.99 \\
  \bottomrule
\end{tabularx}
\caption{Parameter values for the market data example.}
\label{table:params_market}
\end{table}

The OT-calibrated model volatility skews are plotted in Figure \ref{fig:vol_skews_market}, and the simulation of $X$ is given in Figure \ref{fig:sim_x_market}. From the plots, we can see that the OT-calibrated model accurately captures the market data while keeping $X^2_{T,T}=0$ $\bP$-a.s. satisfied. The volatility behaviour is displayed in Appendix \ref{app:beta_market}.

\begin{remark}
Theoretically, the choice of $\bar\beta$ should affect the calibration result, but not the feasibility thereof. If there is only one model that calibrates to the constraints, e.g., when calibrating to option prices with all strikes available, the result will not depend on $\bar\beta$. The degree of freedom in the choice of $\bar\beta$ and the cost function allow us to calibrate a model even when option constraints are sparse.
\end{remark}

\begin{remark}
In Figure \ref{fig:sim_x_market}, we observe a rapid distribution change of $X^2$ after the VIX options expiry. Recall that our $X^2$ is the forward expected quadratic variation of $X^1$, which is indeed the scaled variance swap. Since the market prices are from the true VIX options, this rapid distribution change could be caused by the discrepancy between the VIX value and the variance swap which does not have a listed market. This discrepancy is well known to practitioners. In our approach, the VIX is inferred from the true log-contract, coherently with the variance swap. This approximation could lead to a slight incoherence with observed market prices. We have indeed observed that the convergence of the calibration was highly sensitive to $X^2_0$. Note that the same approach still works if we replace $X^2$ by the combination of vanilla options that is used in the CBOE VIX calculation, which then allows us to potentially get better values of $X^2_0$ from market prices. However, we did not model $X^2$ that way here for the simplicity of presentation.
\end{remark}

\begin{figure}[ht!]
   \begin{minipage}{0.49\textwidth}
     \centering
     \includegraphics[width=1.0\linewidth]{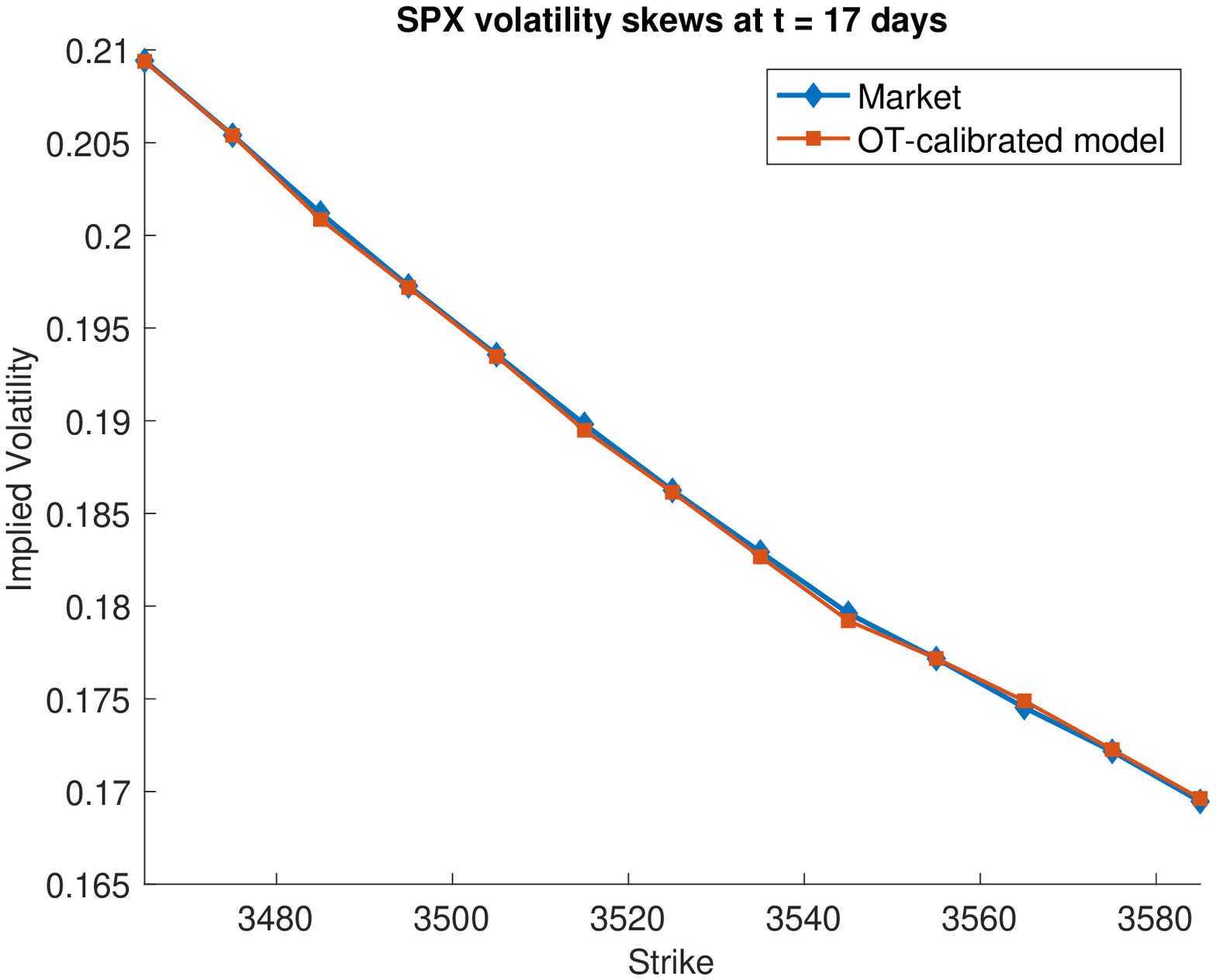}
   \end{minipage}\hfill
   \begin{minipage}{0.49\textwidth}
     \centering
     \includegraphics[width=1.0\linewidth]{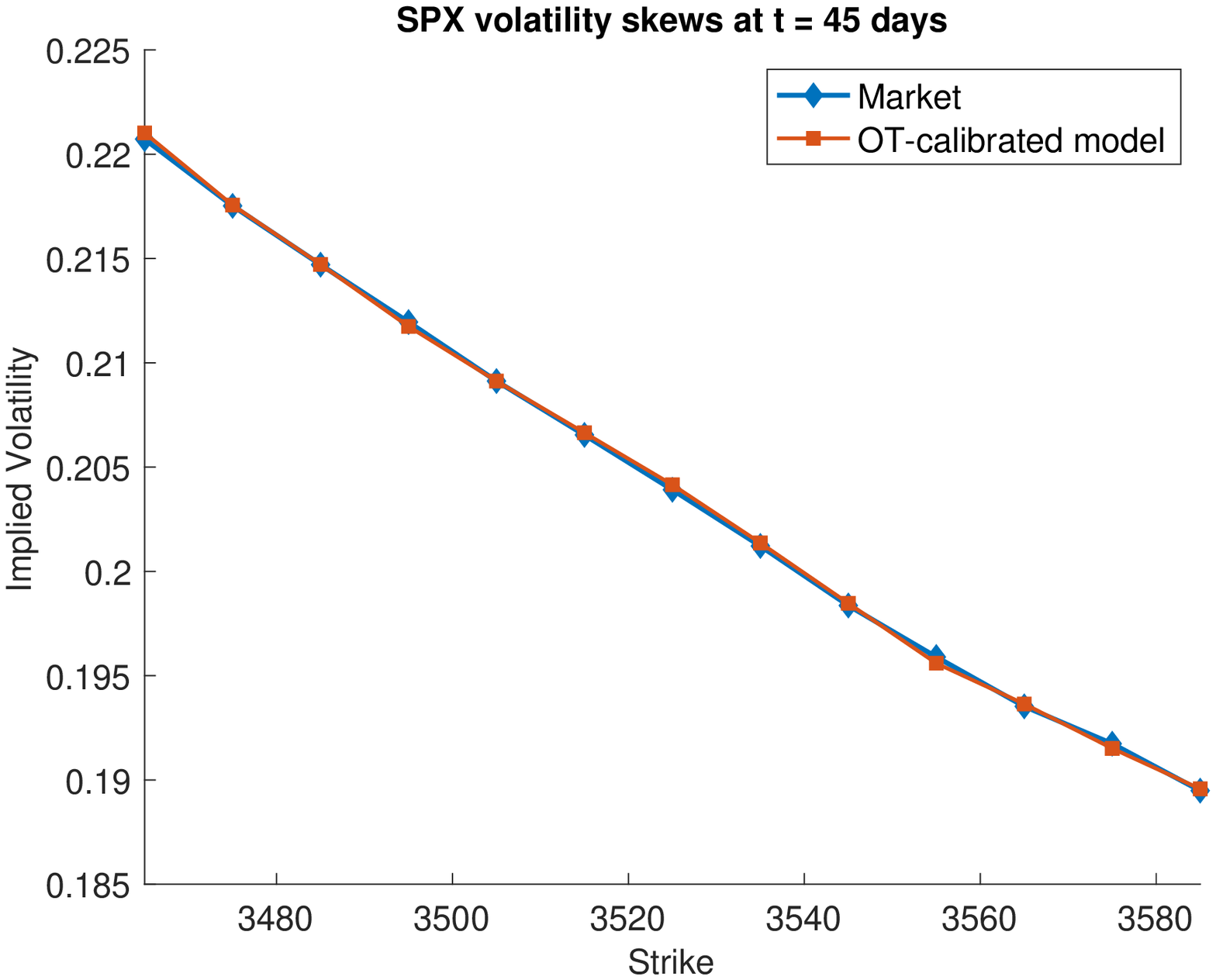}
   \end{minipage}
   \centering
   \begin{minipage}{0.49\textwidth}
     \centering
     \includegraphics[width=1.0\linewidth]{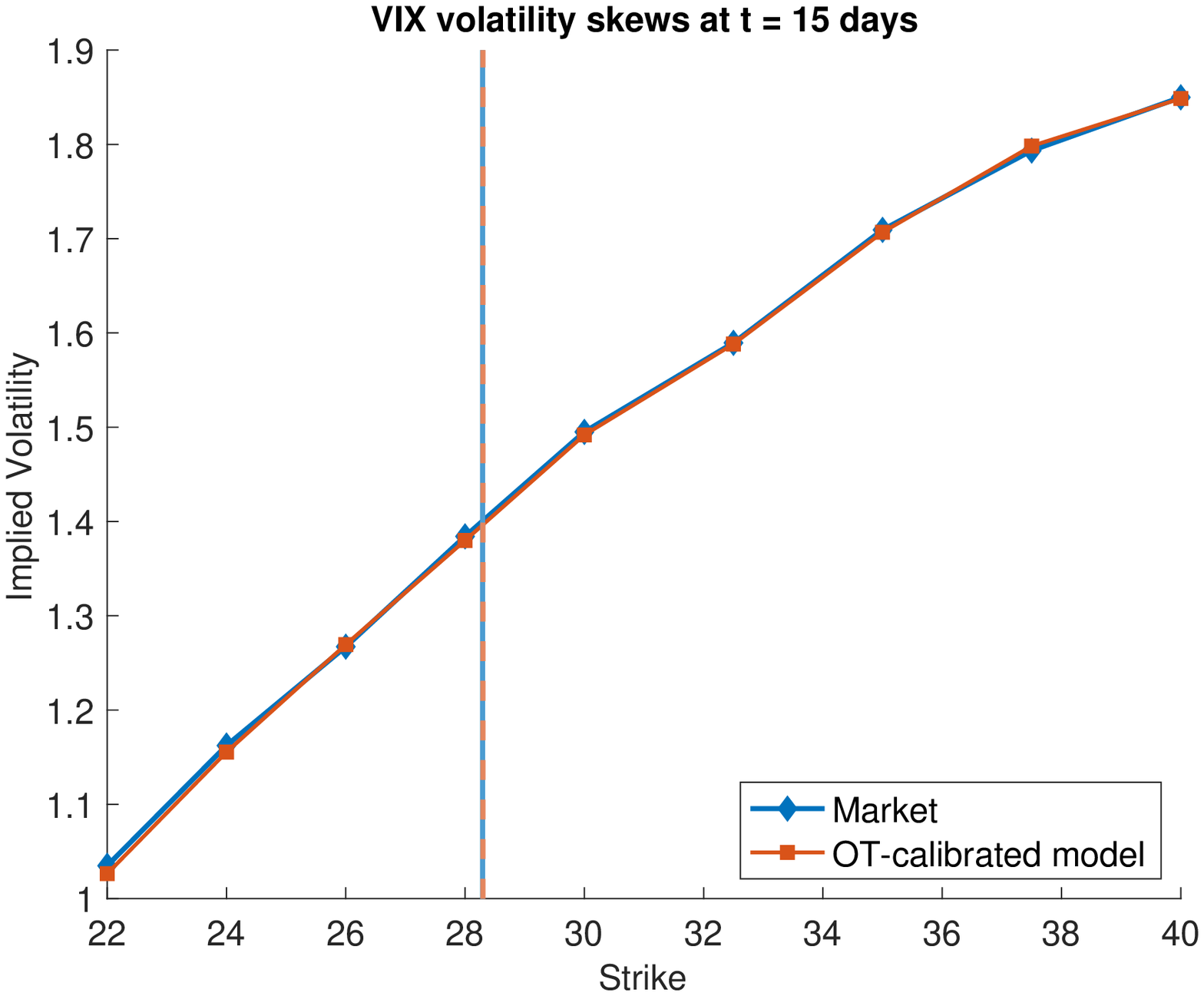}
   \end{minipage}
   \caption{Approximated OT-calibrated model volatility skews of SPX options at $t_0+2$ days $=17$ days, SPX options at $T = 45$ days and VIX options at $t_0=15$ days in the market data example. The vertical lines are VIX futures prices. Markers correspond to computed prices which are then interpolated with a piece-wise linear function.}
   \label{fig:vol_skews_market}
\end{figure}

\begin{figure}[ht!]
\begin{center}
\includegraphics[width=\textwidth]{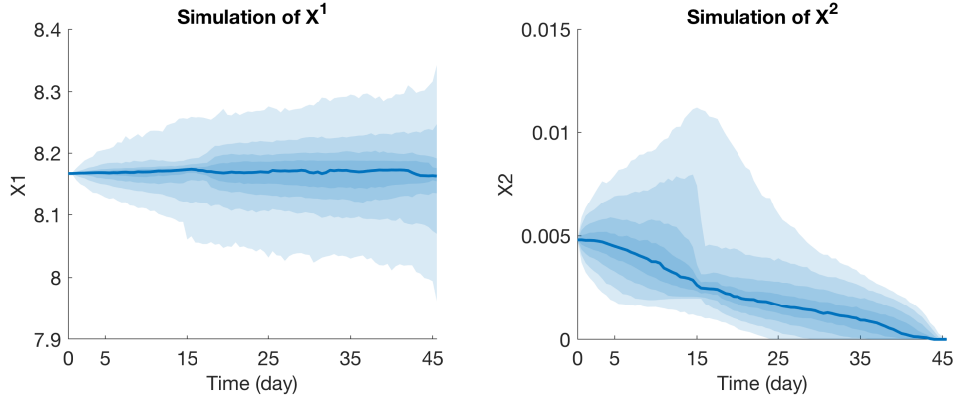}
\caption{The simulations of the OT-calibrated model $X$ in the market data example. }
\label{fig:sim_x_market}
\end{center}
\end{figure}

\section*{Acknowledgements}

The Centre for Quantitative Finance and Investment Strategies has been supported by BNP Paribas. The last author is supported by an Australian Government Research Training Program (RTP) Scholarship. Part of this research was carried out when J. Ob\l\'oj was visiting Monash University and the Sydney Mathematical Research Institute and he gratefully acknowledges their hospitality.

\appendix
\section{The convex conjugate \texorpdfstring{$F^*$}{F*}}
Given $a\in\R^2, b\in\bS^2$ and $\bar\beta\in\bS^2$, define
\eqn{
  A &:= \bar\beta_{11} + \demi b_{11} - \frac{1}{4}a_1 - \frac{1}{4}a_2, \\
  B &:= \bar\beta_{12} + \demi b_{12}, \\ 
  C &:= \bar\beta_{22} + \demi b_{22}, \\
  M &:= \left[ \begin{array}{cc} A & B \\ B & C \end{array}\right].
}
We also define
\eqn{
  x'_+ &:= \frac{A-C}{4} + \frac{A^2-C^2}{4\sqrt{4B^2+(A-C)^2}}, & x'_- &:= \frac{A-C}{4} - \frac{A^2-C^2}{4\sqrt{4B^2+(A-C)^2}}, \\
  y'_+ &:= \frac{B}{2} + \frac{B(A+C)}{2\sqrt{4B^2+(A-C)^2}}, & y'_- &:= \frac{B}{2} - \frac{B(A+C)}{2\sqrt{4B^2+(A-C)^2}},
}
and define
\eqn{
  \lambda_+ &:= \left[ \begin{array}{cc} x'_+ + \sqrt{(x'_+)^2+(y'_+)^2} & y'_+ \\ y'_+ & -x'_+ + \sqrt{(x'_+)^2+(y'_+)^2} \end{array}\right], \\
  \lambda_- &:= \left[ \begin{array}{cc} x'_- + \sqrt{(x'_-)^2+(y'_-)^2} & y'_- \\ y'_- & -x'_- + \sqrt{(x'_-)^2+(y'_-)^2} \end{array}\right].
}
\begin{lemma}\label{lemma:f_conjugate}
The convex conjugate of $F$ is
\eqn{
  F^*(a,b) = (b_{11}-\demi a_1-\demi a_2)\beta^*_{11} + 2b_{12}\beta^*_{12} + b_{22}\beta^*_{22} - \sum_{i,j=1}^2 (\beta^*_{ij} - \bar\beta_{ij})^2,
}
where the values of $\beta^*$ are determined as follows:
\begin{enumerate}
  \item If $M\in\bS^2_+$, then $\beta^* = M$.
  \item If $AC\geq B^2$ and $A+C<0$, then $\beta^*$ is the null matrix.
  \item Otherwise, 
    \eqn{
      \beta^* = \argmin_{\beta\in\{\lambda_+,\lambda_-\}} (\beta_{11}-A)^2 + 2(\beta_{12}-B)^2 + (\beta_{22}-C)^2.
    }
\end{enumerate}
\end{lemma}
\begin{proof}
By definition, the convex conjugate of $F$ is given by
\eqn{
  F^*(a,b) &= \sup_{\beta\in\bS^2_+}\{ -\demi a_1\beta_{11} - \demi a_2\beta_{11} + b_{11}\beta_{11} + 2b_{12}\beta_{12} + b_{22}\beta_{22} - \sum_{i,j=1}^2 (\beta_{ij} - \bar\beta_{ij})^2\} \\
    &= - \inf_{\beta\in\bS^2_+}\{ (\beta_{11} -A)^2 + 2(\beta_{12}-B)^2+(\beta_{22}-C)^2 \} + (A^2-\bar\beta_{11}^2) + 2(B^2-\bar\beta_{12}^2) + (C^2-\bar\beta_{22}^2).
}
Finding the $\beta$ that achieves the above infimum is equivalent to solving
\eq{\label{inf:A1}
  (\beta_{11},\beta_{12},\beta_{22}) = \arginf_{(x,y,z)\in\R_{\geq 0}\times\R\times\R_{\geq 0}}\{ (x-A)^2 + 2(y-B)^2 + (z-C)^2 \mid xz\geq  y^2 \}.
}
In order to solve this problem, let us rotate the $xyz$-axes around $y$-axis clockwise through an angle of $45^{\circ}$ into $x'y'z'$-axes, which can be described by the linear transformation:
\eqn{
  \begin{pmatrix}x' \\ y' \\ z'\end{pmatrix} = \begin{pmatrix}\demi & 0 & -\demi \\ 0 & 1 & 0 \\ \demi & 0 & \demi\end{pmatrix} \begin{pmatrix}x \\ y \\ z\end{pmatrix}.
}
The inverse transformation is
\eqn{
  \begin{pmatrix}x \\ y \\ z\end{pmatrix} = \begin{pmatrix}1 & 0 & 1 \\ 0 & 1 & 0 \\ -1 & 0 & 1\end{pmatrix} \begin{pmatrix}x' \\ y' \\ z'\end{pmatrix}.
}
In terms of $(x',y',z')$, the infimum in (\ref{inf:A1}) can be reformulated as
\eq{\label{inf:A2}
  \inf_{(x',y',z')\in W} 2(x'-\xbar')^2 + 2(y'-\ybar')^2 + 2(z'-\zbar')^2,
}
where $(\xbar', \ybar', \zbar'):=(\demi A-\demi C,B,\demi A+\demi C)$, and $W$ is a convex cone defined as 
\eqn{
  W = \{ (x',y',z')\in\R^3 \mid z'\geq 0,\, x'^2+y'^2 \leq z'^2 \}.
}
In the $x'y'z'$-axes, the above problem can be simply described as finding the minimum Euclidean distance from the point $(\xbar',\ybar',\zbar')$ to $W$. There are three cases:
\begin{enumerate}[label=(\alph*)]
  \item If $(\xbar',\ybar',\zbar')\in W$, the solution is $(x',y',z') = (\xbar',\ybar',\zbar')$.
  \item If $\xbar'^2+\ybar'^2 \leq \zbar'^2$, but $\zbar'<0$. Then the solution should be on the boundary $z'=0$, which also implies that $x'=y'=0$.
  \item Otherwise, the solution must be on the boundary of W:
    \eqn{
      \partial W = \{ (x',y',z')\in\R^3 \mid z'\geq 0,\, x'^2+y'^2 = z'^2 \}.
    }
    By substituting $z' = \sqrt{x'^2+y'^2}$ into (\ref{inf:A2}) and solving the infimum, we find two stationary points:
    \eqn{
      (x'_+,y'_+,z'_+) &= \left(\frac{\xbar'}{2} + \frac{\xbar'\zbar'}{2\sqrt{\xbar'^2+\ybar'^2}}, \frac{\ybar'}{2} + \frac{\ybar'\zbar'}{2\sqrt{\xbar'^2+\ybar'^2}}, \sqrt{(x_+')^2+(y_+')^2} \right), \\
      (x'_-,y'_-,z'_-) &= \left(\frac{\xbar'}{2} - \frac{\xbar'\zbar'}{2\sqrt{\xbar'^2+\ybar'^2}}, \frac{\ybar'}{2} - \frac{\ybar'\zbar'}{2\sqrt{\xbar'^2+\ybar'^2}}, \sqrt{(x_-')^2+(y_-')^2} \right).
    }
    One of the stationary points achieves the infimum. Thus, we choose the one with the smaller objective value.
\end{enumerate} 
Transforming the above solutions back to the $xyz$-axes through the inverse transformation and replacing $(x,y,z)$ by $(\beta_{11},\beta_{12},\beta_{22})$, we obtain the desired result.
\end{proof}

\clearpage\section{Algorithm} \label{appendix:algo}

Let $\pi^N:=\{t_k:0\leq k\leq N\}$ be a discretisation of $[0,T]$ such that $0=t^0<t^1<\ldots<t^{N}=T$. We assume that each of $t_0$ and $\tau_i, i=1,\ldots,m$ coincides with some value in $\pi^N$. Denote by $\epsilon_1$ the tolerance of the maximum of the gradients (\ref{eq:grad_spx})--(\ref{eq:grad_singular}), and denote by $\epsilon_2$ the tolerance for the policy iteration. Recall that $\epsilon_1$ has an alternative interpretation as the tolerance of the maximum error between the calibrating prices and the model prices. In the numerical example presented in Section \ref{sec:num_exp}, $\epsilon_1=10^{-4}$ and $\epsilon_2=10^{-8}$. The numerical method described in Section \ref{sec:num_method} is summarised as the following algorithm.

\begin{algorithm}[!ht]
\DontPrintSemicolon
\SetNoFillComment
\setcounter{AlgoLine}{0}
  Set an initial $(\lambda^{SPX},\lambda^{VIX,f}, \lambda^{VIX},\lambda^\xi)$ \\
  \Do{The maximum of the gradients (\ref{eq:grad_spx}) to (\ref{eq:grad_singular}) is greater than $\epsilon_1$}{
    \tcc{Solving the HJB equation}
    \For{$k = N-1,\ldots,0$}{
      \tcc{Terminal conditions}
      \If{$\exists i = 1,\ldots,m ,\; t_{k+1} = \tau_i $}{
        $\phi_{t_{k+1}} \gets \phi_{t_{k+1}} + \sum_{i=1}^m \lambda_i^{SPX} G_i \mathds{1}(t_{k+1}=\tau_i)$ \tcp*{SPX options}
      }
      \If{$t^{k+1} = t_0$}{
        $\phi_{t^{k+1}} \gets \phi_{t^{k+1}} + \lambda^{VIX,f} J$ \tcp*{VIX futures}
        $\phi_{t^{k+1}} \gets \phi_{t^{k+1}} + \sum_{i=1}^n\lambda_i^{VIX} (H_i\circ J)$ \tcp*{VIX options}
      }
      \If{$t^{k+1} = T$}{
        $\phi_{t^{k+1}} \gets \phi_{t^{k+1}} + \lambda^\xi \xi$ \tcp*{Singular contract}
      }
      \tcc{Policy iteration}
      $\phi_{t_k}^{new} \gets \phi_{t_{k+1}}$ \\
      \Do{$\norm{\phi_{t_k}^{new} - \phi_{t_k}^{old}}_\infty > \epsilon_2$}{
          $\phi_{t_k}^{old} \gets \phi_{t_k}^{new}$ \\
          Approximate $\beta^*$ by Lemma \ref{lemma:f_conjugate} with $\phi_{t_k}^{old}$  \\
          Solve the HJB equation (\ref{hjb:t0_T}) or (\ref{hjb:0_t0}) with $\beta^*$ as a linearised PDE by the standard implicit finite difference method, and set the solution as $\phi_{t_k}^{new}$
      }
      $\phi_{t_k} \gets \phi_{t_{k}}^{new}$
    }
    \tcc{Model prices and gradients}
    Calculate the model prices by solving equations \eqref{pde:pricing} by the ADI method \\
    Calculate the gradients (\ref{eq:grad_spx}) to (\ref{eq:grad_singular})  \\
    Update $(\lambda^{SPX},\lambda^{VIX,f}, \lambda^{VIX},\lambda^\xi)$ by the L-BFGS algorithm
  }
 \caption{The joint calibration algorithm}
 \label{algo}
\end{algorithm}

\clearpage\section{The diffusion process \texorpdfstring{$\beta$}{Beta} for the simulated data example}
\label{app:beta}

\begin{figure}[!ht]
\begin{center}
\includegraphics[width=0.95\textwidth]{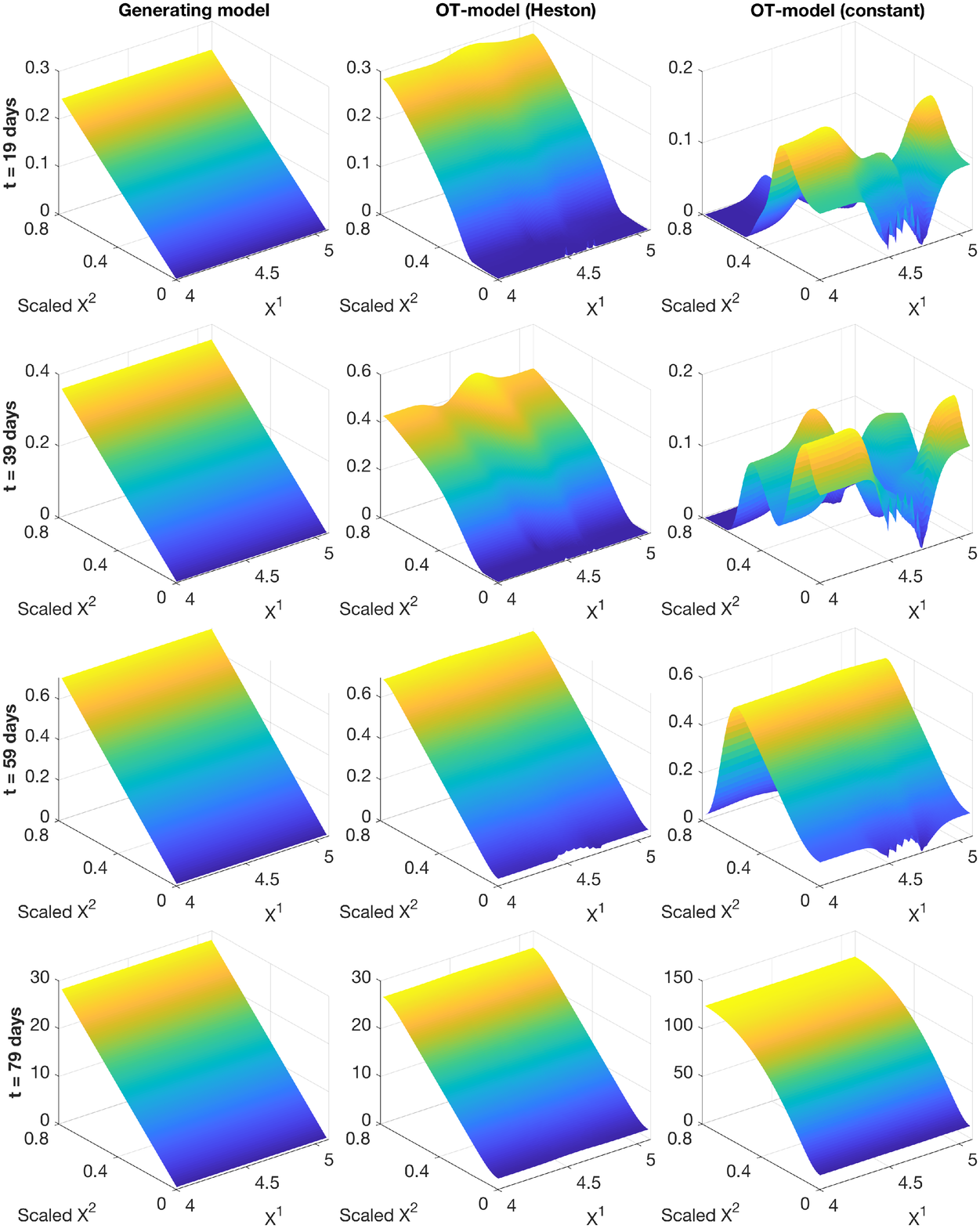}
\caption{The functions $\beta_{11}(t,X^1,X^2)$ of the generating model, the OT-calibrated model with a Heston reference and the OT-calibrated model with a constant reference for the simulated data example.}
\end{center}
\end{figure}

\clearpage
\begin{figure}[!ht]
\begin{center}
\includegraphics[width=1.0\textwidth]{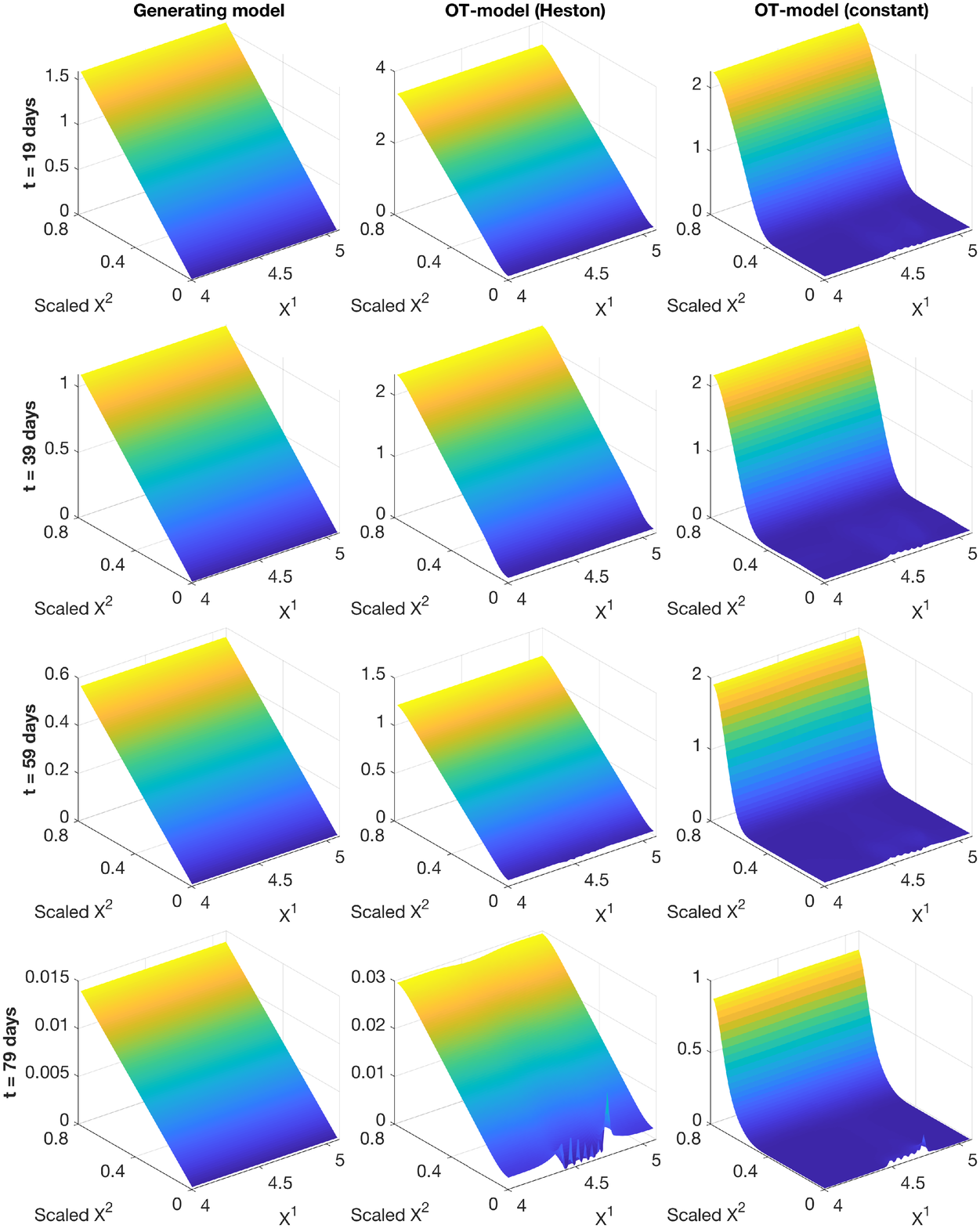}
\caption{The functions $\beta_{22}(t,X^1,X^2)$ of the generating model, the OT-calibrated model with a Heston reference and the OT-calibrated model with a constant reference for the simulated data example.}
\end{center}
\end{figure}

\clearpage
\begin{figure}[!ht]
\begin{center}
\includegraphics[width=1.0\textwidth]{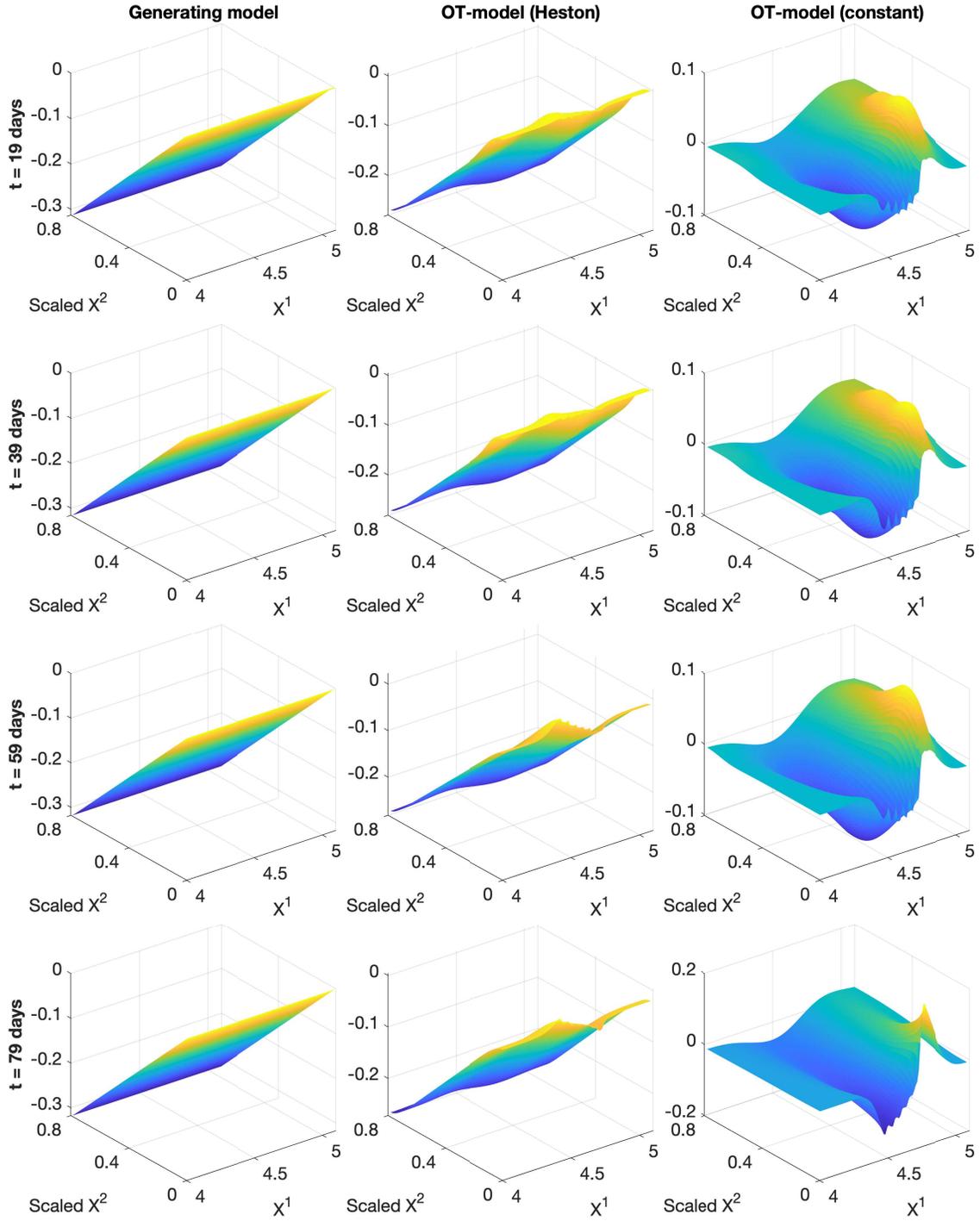}
\caption{The functions $\beta_{12}(t,X^1,X^2)$ of the generating model, the OT-calibrated model with a Heston reference and the OT-calibrated model with a constant reference for the simulated data example.}
\end{center}
\end{figure}

\clearpage\section{The diffusion process \texorpdfstring{$\beta$}{Beta} for the market data example}
\label{app:beta_market}

\begin{figure}[!ht]
\begin{center}
\includegraphics[width=0.95\textwidth]{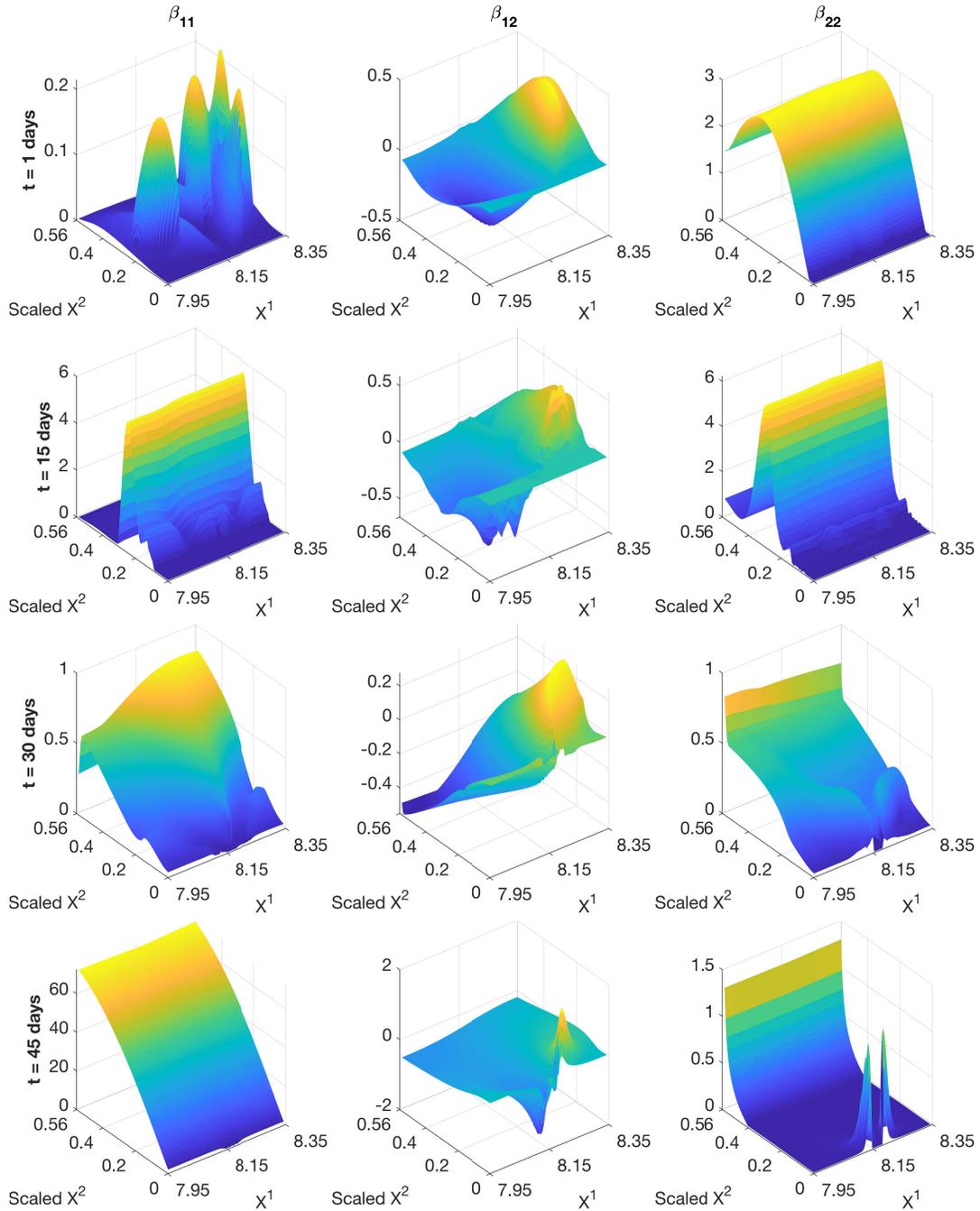}
\caption{The functions $\beta_{11}(t,X^1,X^2)$, $\beta_{12}(t,X^1,X^2)$ and $\beta_{22}(t,X^1,X^2)$ of the OT-calibrated model for the market data example.}
\end{center}
\end{figure}

\clearpage
\bibliographystyle{acm}
\bibliography{ref}

\begin{thebibliography}{10}

\bibitem{acciaio2020inversion}
{\sc Acciaio, B., and Guyon, J.}
\newblock Inversion of convex ordering: Local volatility does not maximize the
  price of {VIX} futures.
\newblock {\em SIAM J. Financial Math. 11}, 1 (2020), SC--1.

\bibitem{Andreev2017mfg}
{\sc Andreev, R.}
\newblock Preconditioning the augmented {L}agrangian method for instationary
  mean field games with diffusion.
\newblock {\em SIAM J. Sci. Comput. 39}, 6 (2017), A2763--A2783.

\bibitem{avellaneda1997calibrating}
{\sc Avellaneda, M., Friedman, C., Holmes, R., and Samperi, D.}
\newblock Calibrating volatility surfaces via relative-entropy minimization.
\newblock {\em Appl. Math. Finance 4}, 1 (1997), 37--64.

\bibitem{baldeaux2014consistent}
{\sc Baldeaux, J., and Badran, A.}
\newblock Consistent modelling of {VIX} and equity derivatives using a 3/2 plus
  jumps model.
\newblock {\em Appl. Math. Finance 21}, 4 (2014), 299--312.

\bibitem{Barles1991monotone}
{\sc Barles, G., and Souganidis, P.~E.}
\newblock Convergence of approximation schemes for fully nonlinear second order
  equations.
\newblock {\em Asymptotic Anal. 4}, 3 (1991), 271--283.

\bibitem{BHLP:13}
{\sc Beiglb{\"o}ck, M., Henry-Labord{\`e}re, P., and Penkner, F.}
\newblock {Model-independent bounds for option prices---a mass transport
  approach}.
\newblock {\em Finance Stoch. 17}, 3 (2013), 477--501.

\bibitem{benamou-brenier2000}
{\sc Benamou, J.-D., and Brenier, Y.}
\newblock A computational fluid mechanics solution to the
  {M}onge--{K}antorovich mass transfer problem.
\newblock {\em Numer. Math. 84}, 3 (2000), 375--393.

\bibitem{Bonnans2003hjb}
{\sc Bonnans, J.~F., and Zidani, H.}
\newblock Consistency of generalized finite difference schemes for the
  stochastic {HJB} equation.
\newblock {\em SIAM J. Numer. Anal. 41}, 3 (2003), 1008--1021.

\bibitem{brunick2013}
{\sc Brunick, G., and Shreve, S.}
\newblock Mimicking an {I}t\^o process by a solution of a stochastic
  differential equation.
\newblock {\em Ann. Appl. Probab. 23}, 4 (2013), 1584--1628.

\bibitem{carr1998volatility}
{\sc Carr, P., and Madan, D.}
\newblock Towards a theory of volatility trading.
\newblock In {\em Volatility}, R.~Jarrow, Ed. Risk Publications, 1998,
  pp.~417--27.

\bibitem{cont2013consistent}
{\sc Cont, R., and Kokholm, T.}
\newblock A consistent pricing model for index options and volatility
  derivatives.
\newblock {\em Math. Finance 23}, 2 (2013), 248--274.

\bibitem{davis2008completeness}
{\sc Davis, M., and Ob\l\'{o}j, J.}
\newblock Market completion using options.
\newblock In {\em Advances in mathematics of finance}, vol.~83 of {\em Banach
  Center Publ.} Polish Acad. Sci. Inst. Math., Warsaw, 2008, pp.~49--60.

\bibitem{de2015linking}
{\sc De~Marco, S., and Henry-Labordere, P.}
\newblock Linking vanillas and {VIX} options: a constrained martingale optimal
  transport problem.
\newblock {\em SIAM J. Financial Math. 6}, 1 (2015), 1171--1194.

\bibitem{Debrabant2013semilagrangian}
{\sc Debrabant, K., and Jakobsen, E.~R.}
\newblock Semi-{L}agrangian schemes for linear and fully non-linear diffusion
  equations.
\newblock {\em Math. Comp. 82}, 283 (2013), 1433--1462.

\bibitem{dupire1993arbitrage}
{\sc Dupire, B.}
\newblock Arbitrage pricing with stochastic volatility.
\newblock {\em Preprint\/} (1993).

\bibitem{Figalli2008existence}
{\sc Figalli, A.}
\newblock Existence and uniqueness of martingale solutions for {SDE}s with
  rough or degenerate coefficients.
\newblock {\em J. Funct. Anal. 254}, 1 (2008), 109--153.

\bibitem{fouque2018heston}
{\sc Fouque, J.-P., and Saporito, Y.~F.}
\newblock Heston stochastic vol-of-vol model for joint calibration of {VIX} and
  {S\&P} 500 options.
\newblock {\em Quant. Finance 18}, 6 (2018), 1003--1016.

\bibitem{HenryLabordere:2014hta}
{\sc Galichon, A., Henry-Labord\`ere, P., and Touzi, N.}
\newblock A stochastic control approach to no-arbitrage bounds given marginals,
  with an application to lookback options.
\newblock {\em Ann. Appl. Probab. 24}, 1 (2014), 312--336.

\bibitem{gatheral2008consistent}
{\sc Gatheral, J.}
\newblock Consistent modeling of {SPX} and {VIX} options.
\newblock In {\em Bachelier congress\/} (2008), vol.~37, pp.~39--51.

\bibitem{gatheral2020quadratic}
{\sc Gatheral, J., Jusselin, P., and Rosenbaum, M.}
\newblock The quadratic rough {H}eston model and the joint {S\&P} 500/{VIX}
  smile calibration problem.
\newblock {\em Risk, May\/} (2020).

\bibitem{goutte2017regime}
{\sc Goutte, S., Ismail, A., and Pham, H.}
\newblock Regime-switching stochastic volatility model: estimation and
  calibration to {VIX} options.
\newblock {\em Appl. Math. Finance 24}, 1 (2017), 38--75.

\bibitem{guo2018pricing}
{\sc Guo, I., and Loeper, G.}
\newblock Pricing bounds for volatility derivatives via duality and least
  squares {M}onte {C}arlo.
\newblock {\em J. Optim. Theory Appl. 179}, 2 (2018), 598--617.

\bibitem{guo2018path}
{\sc Guo, I., and Loeper, G.}
\newblock Path dependent optimal transport and model calibration on exotic
  derivatives.
\newblock {\em Ann. Appl. Probab. 31}, 3 (2021), 1232--1263.

\bibitem{guo2017local}
{\sc Guo, I., Loeper, G., and Wang, S.}
\newblock Local volatility calibration by optimal transport.
\newblock In {\em 2017 {MATRIX} annals}, vol.~2 of {\em MATRIX Book Ser.}
  Springer, Cham, 2019, pp.~51--64.

\bibitem{guo2019calibration}
{\sc Guo, I., Loeper, G., and Wang, S.}
\newblock Calibration of local-stochastic volatility models by optimal
  transport.
\newblock {\em Math. Finance 31\/} (2021).

\bibitem{guyon2020inversion}
{\sc Guyon, J.}
\newblock Inversion of convex ordering in the {VIX} market.
\newblock {\em Quant. Finance\/} (2020), 1--27.

\bibitem{guyon2020joint}
{\sc Guyon, J.}
\newblock The joint {S\&P 500/VIX} smile calibration puzzle solved.
\newblock {\em Risk, April\/} (2020).

\bibitem{gyongy1986mimicking}
{\sc Gy\"{o}ngy, I.}
\newblock Mimicking the one-dimensional marginal distributions of processes
  having an {I}t\^{o} differential.
\newblock {\em Probab. Theory Relat. Fields 71}, 4 (1986), 501--516.

\bibitem{heston1993closed}
{\sc Heston, S.~L.}
\newblock A closed-form solution for options with stochastic volatility with
  applications to bond and currency options.
\newblock {\em Rev. Financ. Stud. 6}, 2 (1993), 327--343.

\bibitem{foulon2010adi}
{\sc In~'t Hout, K.~J., and Foulon, S.}
\newblock A{DI} finite difference schemes for option pricing in the {H}eston
  model with correlation.
\newblock {\em Int. J. Numer. Anal. Model. 7}, 2 (2010), 303--320.

\bibitem{jacquier2018vix}
{\sc Jacquier, A., Martini, C., and Muguruza, A.}
\newblock On {VIX} futures in the rough {B}ergomi model.
\newblock {\em Quant. Finance 18}, 1 (2018), 45--61.

\bibitem{kokholm2015joint}
{\sc Kokholm, T., and Stisen, M.}
\newblock Joint pricing of {VIX} and {SPX} options with stochastic volatility
  and jump models.
\newblock {\em J. Risk Finance\/} (2015).

\bibitem{Kushner2001numerical}
{\sc Kushner, H.~J., and Dupuis, P.}
\newblock {\em Numerical methods for stochastic control problems in continuous
  time}, second~ed., vol.~24 of {\em Applications of Mathematics (New York)}.
\newblock Springer-Verlag, New York, 2001.

\bibitem{LBFGS1989}
{\sc Liu, D.~C., and Nocedal, J.}
\newblock On the limited memory {BFGS} method for large scale optimization.
\newblock {\em Math. Programming 45}, 3, (Ser. B) (1989), 503--528.

\bibitem{loeper2006}
{\sc Loeper, G.}
\newblock The reconstruction problem for the {E}uler-{P}oisson system in
  cosmology.
\newblock {\em Arch. Ration. Mech. Anal. 179}, 2 (2006), 153--216.

\bibitem{ma2017monotone}
{\sc Ma, K., and Forsyth, P.~A.}
\newblock An unconditionally monotone numerical scheme for the two-factor
  uncertain volatility model.
\newblock {\em IMA J. Numer. Anal. 37}, 2 (2017), 905--944.

\bibitem{neuberger1994log}
{\sc Neuberger, A.}
\newblock The log contract.
\newblock {\em Journal of portfolio management 20\/} (1994), 74--74.

\bibitem{pacati2018smiling}
{\sc Pacati, C., Pompa, G., and Ren{\`o}, R.}
\newblock Smiling twice: The {H}eston++ model.
\newblock {\em J. Bank. Finance 96\/} (2018), 185--206.

\bibitem{papanicolaou2014regime}
{\sc Papanicolaou, A., and Sircar, R.}
\newblock A regime-switching heston model for {VIX} and {S\&P} 500 implied
  volatilities.
\newblock {\em Quant. Finance 14}, 10 (2014), 1811--1827.

\bibitem{song2012tale}
{\sc Song, Z., and Xiu, D.}
\newblock A tale of two option markets: State-price densities implied from
  {S\&P} 500 and {VIX} option prices.
\newblock {\em Unpublished working paper. Federal Reserve Board and University
  of Chicago\/} (2012).

\bibitem{tan-touzi2013}
{\sc Tan, X., and Touzi, N.}
\newblock Optimal transportation under controlled stochastic dynamics.
\newblock {\em Ann. Probab. 41}, 5 (2013), 3201--3240.

\bibitem{Trevisan2016existence}
{\sc Trevisan, D.}
\newblock Well-posedness of multidimensional diffusion processes with weakly
  differentiable coefficients.
\newblock {\em Electron. J. Probab. 21\/} (2016), Paper No. 22, 41.

\bibitem{villani2003book}
{\sc Villani, C.}
\newblock {\em Topics in optimal transportation}, vol.~58 of {\em Graduate
  Studies in Mathematics}.
\newblock American Mathematical Society, Providence, RI, 2003.

\end{thebibliography}

\end{document}